\documentclass[letterpaper,onecolumn,10pt, accepted=2026-06-28]{quantumarticle}
\pdfoutput=1
\usepackage[utf8]{inputenc}
\usepackage[english]{babel}
\usepackage[T1]{fontenc}
\usepackage{hyperref}
\usepackage[numbers,sort&compress]{natbib}
\usepackage{lipsum}

\usepackage{tikz}
\usetikzlibrary{%
    decorations.pathreplacing,%
    decorations.pathmorphing,%
    arrows
}
\usetikzlibrary{positioning, calc,intersections,shapes,matrix,fit}

\usepackage{float}
\usepackage{graphicx}
\usepackage{dcolumn}
\usepackage{bm}
\usepackage{appendix}
\usepackage{subcaption}
\usepackage{soul}
\usepackage{multirow}
\usepackage{booktabs}
\usepackage{amsmath}
\usepackage{amsfonts}
\usepackage[braket, qm]{qcircuit}
\usepackage[normalem]{ulem}
\usepackage{url}
\usepackage{array}
\usepackage{titlesec}
\usepackage{dsfont}
\usepackage{amsthm}
\usepackage{thmtools}

\def\be{\begin{equation}}
\def\ee{\end{equation}}
                               
\usepackage{etoolbox}
\makeatletter
\patchcmd{\@caption}{\raggedright}{}{}{}
\makeatother

\newcommand{\tr}[1]{\mathrm{tr} \left( #1 \right)}
\renewcommand{\ket}[1]{| #1 \rangle}

\newcommand{\kett}[1]{| \,#1\, \rangle\kern-0.2em\rangle}
\newcommand{\braa}[1]{\langle\kern-0.2em\langle \,#1\, |}

\renewcommand{\eqref}[1]{\textsc{eq.}~(\ref{#1})}
\newcommand{\secref}[1]{\S\,\ref{#1}}
\newcommand{\partl}[3]{\frac{\partial^{#3}#1}{\partial#2^{#3}}}

\newcommand{\afchem}{Department of Chemistry, University of California, Berkeley, California 94720, USA}
\newcommand{\afbqic}{Berkeley Center for Quantum Information and Computation, Berkeley, California 94720, USA}
\newcommand{\inria}{Laboratoire de Physique de l’\'{E}cole Normale Sup\'{e}rieure Inria, ENS, Mines ParisTech, Universit\'{e} PSL, Sorbonne Universit\'{e} Paris France}
\newcommand{\afphys}{Department of Physics, University of California, Berkeley, California 94720, USA}

\begin{document}

\title{Optimal schedule of multi-channel quantum Zeno dragging with application to solving the $k$-SAT problem}

\author{Yipei Zhang}
\email{zyp123@berkeley.edu}
\affiliation{\afbqic}
\affiliation{\afchem}
\author{Alain Sarlette}
\email{alain.sarlette@inria.fr}
\affiliation{\inria}
\author{Philippe Lewalle}
\affiliation{\afbqic}
\affiliation{\afchem}
\email{plewalle@berkeley.edu}
\author{Tathagata Karmakar}
\email{tatha@berkeley.edu}
\affiliation{\afbqic}
\affiliation{\afchem}
\author{Thilo Scharnhorst}
\email{thilo@berkeley.edu}
\affiliation{\afbqic}
\affiliation{\afchem}
\affiliation{\afphys}
\author{K. Birgitta Whaley}
\email{whaley@berkeley.edu}
\affiliation{\afbqic}
\affiliation{\afchem}

\begin{abstract}
Quantum Zeno dragging enables the preparation of common eigenstates of a set of observables by frequent measurement and adiabatic-like modulation of the measurement basis. In this work, we present a deeper analysis of multi-channel Zeno dragging using generalized measurements, i.e.~simultaneously measuring a set of non-commuting observables that vary slowly in time, to drag the state towards a target subspace. For concreteness, we will focus on a measurement--driven approach to solving $k$-SAT problems as examples. 
We first compute some analytical upper bounds on the convergence time, including the effect of finite measurement time resolution. We then apply optimal control theory to obtain the optimal dragging schedule that lower bounds the convergence time, for low-dimensional settings. This study provides a theoretical foundation for multi-channel Zeno dragging and its optimization, and also serves as a guide for designing optimal dragging schedules for quantum information tasks including measurement-driven quantum algorithms.
\end{abstract}

\maketitle

\section{Introduction}
Techniques of quantum control lie at the heart of the development of quantum information processing  technologies over the last three decades, enabling precise manipulation of quantum systems toward tasks such as computation, simulation, as well as sensing and metrology \cite{quantum_control_review1, DAllessandroBook, PMP_tutorial, QSL_review, Koch_2016}. On the one hand, one of the goals of quantum control is to mitigate the effect of unwanted noise sources stemming from coupling to the environment; on the other hand, dissipative quantum control, which essentially engineers the system-environment interaction, has emerged as a novel tool for control and verification of quantum technologies \cite{harrington2022engineered}. 

Among these dissipative quantum control strategies, the quantum Zeno effect (QZE) \cite{misra1977zeno}, in which frequent measurements freeze the Hamiltonian evolution, stands out as a powerful method to stabilize quantum states, subspaces, and dynamics. This has been applied to manipulation of quantum dynamics \cite{Gambetta_2008, Ruskov_2006,raimond2010phase, raimond2012quantum, signoles2014confined, SHG_prl}, creating entangled states and entangling gates \cite{verstraete2009quantum,  Comb_comp_meas_based_stab_Liu,   Scalable_Dissipative_Preparation_Reiter, reiter2017dissipative, PhysRevA_LeoZhou, Doucet_2020, Brown_2022, doucet2023scalable, blumenthal2022demonstration, Lewalle2023multiqubitquantum, lewalle2024prxq}, logical gates on bosonic codes \cite{mirrahimi2014dynamically, touzard2018coherent, Guillaud_2019, gautier2023designing}, and to create or support quantum algorithms \cite{benjamin2017measurement, BZF_PRA, 8104079, cubitt2023dissipative, zhang2024_ksat, Berwald2024groverspeedupfrom, Berwald_2025, PhysRevA.66.032314}. The concept of QZE is 
closely related to Zeno dragging, where the state is pulled along as the measurement basis changes over time \cite{Aharonov_1980, raimond2010phase, SHG_prl, Guillaud_2019, lewalle2024prxq, zhang2024_ksat, BZF_PRA, benjamin2017measurement, Shea_2024}. 
Zeno dragging has a parallel in the approach of adiabatic quantum computation and adiabatic control, where state transfer is achieved through slow modulation of a system's Hamiltonian $\hat{H}(\theta)$ through some control parameter $\theta(t)$  \cite{RevModPhys_adiabatic, Guery_STA-RMP, 10.1063/1.2995837, RevModPhys.annealing}. In contrast, for Zeno dragging the dynamics are generated by some quantum measurement channel $\mathcal{T}(\theta)$ that may be considered as a particular instance of an adiabatic superoperator \cite{Sarandy_2005, Vacanti_2014, Campos_2016}. This Zeno dragging framework is particularly attractive for preparing ground states of frustration-free Hamiltonians, where commonly each term is local and can be measured separately.

In algorithmic applications of the Zeno effect, and of Zeno dragging in particular, measurement-driven algorithms based on the projective QZE have been proposed to 
solve Boolean satisfiability ($k$-SAT) problems and ground state preparation problems of frustration--free Hamiltonians \cite{benjamin2017measurement, BZF_PRA}. The former has recently been further explored by several of the current authors using the more general quantum Zeno setting provided by generalized measurements \cite{zhang2024_ksat}. Empirical results in \cite{zhang2024_ksat} suggest that the limit of weak continuous measurement provides optimal algorithmic performance for a quantum Zeno driven $k$-SAT computation, but rigorous bounds on the effect of measurement strength on algorithmic performance have not yet been established.

In quantum annealing or adiabatic control and computation problems, it is well known that optimizing the schedule according to which the Hamiltonian is modulated over time is important in the overall performance \cite{Roland_Cerf, Brif_2014, genetic, dalgaard2020global, PhysRevA.102.012614, chen2022optimizing}. 
Similarly, for Zeno dragging we can expect that optimizing the schedule of the control parameter $\theta(t)$ which determines the measurement basis as a function of time could also play a significant role in obtaining computationally efficient measurement-driven protocols and algorithms. 
Indeed, \cite{benjamin2017measurement} already empirically proposed a nonlinear schedule that would perform much better than a linear schedule as a function of time. More recently, in \cite{lewalle2024prxq} several of the current authors have provided a theoretical framework based on optimal control that  enables schedule optimization under Zeno dragging with the target of generating a known target state. 

This method, based on interpreting the action functional from the Chantasri Dressel Jordan (CDJ) \cite{CDJ_origin, CJ} stochastic path integral as a cost function for Pontryagin--style optimal control, has been applied to small--scale systems using a single measurement channel \cite{kokaew2024, lewalle2024prxq, karmakar2025cdj}. We refer to this approach as the CDJ--Pontryagin (CDJ--P) method.
The case of multi-channel Zeno dragging, where multiple measurement observables are either randomly selected to be measured in the discrete case or simultaneously monitored with continuous-time signals, holds broader potential. 
In algorithmic applications, the target state reached at the end of the Zeno dragging procedure is generally unknown a priori, because it encodes the unknown solution to the computational problem under consideration \cite{zhang2024_ksat, benjamin2017measurement, BZF_PRA}. 
In this case, the optimal schedule $\theta(t)$, which can be obtained with optimal control theory  based on a full system analysis, provides an upper bound on the performance achievable with a realistic ``optimal'' schedule for the algorithmic application. 

In this paper, we address several of the above-mentioned questions. Specifically, we first prove a version of the adiabatic theorem for multi-channel Zeno dragging with discrete-time measurement of arbitrary strength. 
Indeed, in the time-continuous limit, i.e., with measurement outcomes obtained from continuous-time signals, we are able to weakly and simultaneously measure several non-commuting observables, whose common 0-eigenspace encodes the solution. 
We show that the convergence is guaranteed if one changes $\theta(t)$ slowly enough, compared to the decay of coherences between the solution space and the non-solution space of the ($\theta(t)$-dependent) observables being measured. 
This criterion properly captures the effect of time resolution, promising faster convergence when continuously observing the output signal, as empirically observed in \cite{zhang2024_ksat}.

We then switch to a perspective of optimal control theory to improve the performance in this weak continuous measurement limit, considering optimization of the schedule $\theta(t)$ using first the average (Lindblad) dynamics, and then using the measurement--conditioned dynamics via the CDJ--P formalism. We show that we are able to numerically optimize the measurement basis schedule $\theta(t)$ under arbitrary running and terminal costs, which we demonstrate here with multiple measurement channels for $k$-SAT problems on few-qubit systems. Our numerical results show that such schedule optimization can significantly improve the performance of Zeno dragging for solving quantum algorithms. 

We also note that the idea of dissipative and measurement-driven quantum algorithms has been explored theoretically in several previous works in the literature \cite{verstraete2009quantum, BZF_PRA, PhysRevA.66.032314}. However, these approaches differ from the current work in that they focus either on preparing the ground state of a general Hamiltonian with the associated requirement of implementing the measurement globally \cite{PhysRevA.66.032314}, or on local projective measurements without considering the duration of the measurement process itself \cite{verstraete2009quantum, BZF_PRA}. We will see that the bottleneck governing the overall convergence time is the high number of iterations of the measurement process. Therefore, taking into account the measurement duration is essential to obtain meaningful bounds. We emphasize that the approach presented in the present work considers both varying measurement strength and local measurements for frustration-free operators. We will show that our theoretical results provide a convergence bound in Theorem \ref{theorem_convergence} that improves on the previous results with projective measurements in \cite{BZF_PRA} and also provides a better convergence bound than \cite{BZF_PRA, PhysRevA.66.032314} in terms of the dependence on the spectral gap.

The remainder of this paper is organized as follows. Section \ref{sec_preliminaries} introduces the $k$-SAT setting within the more general context of finding ground states of frustration-free Hamiltonians and defines the Zeno dragging model studied in this work.
The conditional and unconditional dynamics resulting from this measurement-driven paradigm are presented, and the relation to analog quantum computing is clarified. A contrasting implementation within the paradigm of digital (gate--based) quantum computing, that uses entangling gates and ancilla qubits for generalized mid-circuit measurements, is also presented there. 
In Section \ref{sec_convergence}, we then perform a theoretical analysis establishing sufficient convergence conditions for the Zeno dragging implementations. Section \ref{sec_OCT} reviews the relevant concepts of optimal control theory and builds the framework for implementing schedule optimization of the measurement axis parameter $\theta(t)$ in Zeno dragging. We then demonstrate the schedule optimization numerically in Section \ref{sec_numerics}. Finally, we discuss possible future directions and conclude in Section \ref{sec_discussion}. Some details about all of these points can be found in the Appendices.

\section{Defining the Setting for Zeno Dragging and $k$-SAT}\label{sec_preliminaries}

We are interested in preparing ground states of a frustration-free Hermitian operator $\hat{O}_p  =  \frac{1}{m} \sum_{\alpha=1}^m\hat{P}_{\alpha}$ on a (finite-dimensional) Hilbert space. Typically the $\hat{P}_\alpha$ are few-body Hamiltonians, and in particular here we will assume that each of these is a projector involving a few qubits. 
Any ground state of a \emph{frustration-free} Hermitian operator is the simultaneous ground state (here the simultaneous $0$-eigenstate) of all the $\{ \hat{P}_{\alpha}\}_{\alpha = 1}^m$. 
Note that by decomposing $\hat{P}_\alpha = \sum_{\alpha'} c_{\alpha'} \hat{P}_{\alpha,\alpha'}$ with projectors $\hat{P}_{\alpha,\alpha'}$ and positive scalars $c_{\alpha'}$, the setting can be easily generalized to $\hat{P}_\alpha$ being general positive semidefinite operators (not necessarily projectors). 
It is always possible to introduce a control parameter $\theta$, which runs from $\theta_i$ to $\theta_f$, into the set of projectors so that the running problem cost operator
\begin{equation}\label{eq:cost_op}
    \hat{O}(\theta) = \frac{1}{m}\sum_{\alpha}\hat{P}_{\alpha}(\theta)
\end{equation}
remains frustration-free along the path of $\theta$ \cite{BZF_PRA}. 
The final operator $\hat{O}(\theta_f) = \hat{O}_p$ is designed to encode the cost operator of the original problem. It is evident from the definitions of $\hat{O}_p$ and $\hat{P}_\alpha$ that the spectral gap of this final, problem cost operator is an integer multiple of $1/m$, if $\hat{P}_\alpha$ mutually commute, e.g. in the case of classical problems. 
When the initial projector $\hat{O}(\theta_i)$ is designed to have an easy-to-prepare ground state, 
the frustration-free ground state can then in principle be obtained by repeatedly measuring all of the projectors $\hat{P}_\alpha$ while smoothly varying the axis $\theta$ of 
$\hat{P}_\alpha$ from $\theta_i$ to $\theta_f$. 

One example of such $\{\hat{P}_{\alpha}(\theta) \}_{\alpha=1}^m$ was proposed by Benjamin, Zhao and Fitzsimons (BZF) in \cite{benjamin2017measurement, BZF_PRA} to solve classical Boolean satisfiability (SAT) problems, which constitute a type of structured search problem; while many of our results presented below apply more generally, we shall use the class of $k$-SAT problems, where there are $m$ Boolean clauses and each clause interrogates $k$ bits, for illustrative examples.
Appendix \ref{appendix:ksat} provides a brief review of $k$-SAT problems. $k$-SAT problems for $k\geq 3$ are known to be NP-complete \cite{Cook_1971, Karp_1972, Levin_1973}. In \cite{benjamin2017measurement, BZF_PRA}, each constraint (or clause) $\alpha$ of a given instance of the $k$-SAT problem is mapped onto a projector $\hat{P}_{\alpha}$. When the clause $\alpha$ is satisfied, the corresponding projector $\hat{P}_{\alpha}$ is evaluated to give outcome $0$. 
In this way, a simultaneous $0$-eigenstate of all $\hat{P}_{\alpha}$ provides an assignment satisfying all clauses in the original $k$-SAT problem, and thus a solution, if there is any. Specifically, in the projective measurement setting of BZF, the Hilbert space corresponds to $n$ qubits, while each projector $\hat{P}_{\alpha}(\theta)$ only acts on $k$ qubits and corresponds to a clause in a $k$-SAT problem instance:
\begin{equation}\label{eq:projector}
    \hat{P}_{\alpha}(\theta) = \bigotimes_{i=1}^k |l_{\alpha_i}\theta\rangle \langle l_{\alpha_i}\theta|
\end{equation}
with
\begin{equation}\label{eq:projector2}
    |l_{\alpha_i}\theta\rangle = \hat{R}_y(\pi + l_{\alpha_i}\theta)\; |+\rangle_{\alpha_i} \;\; , \quad \text{for}\quad \ket{+} \equiv \tfrac{1}{\sqrt{2}}\left(\ket{0} + \ket{1} \right),
\end{equation}
where $\hat{R}_y(\phi)$ is the single-qubit rotation operator by an angle  $\phi$ around the $y$ axis of the Bloch sphere. 
The $\alpha_i \in [n]$ label the (qu)bits involved in the clause $\alpha$, while the $\{l_{\alpha_1}, l_{\alpha_2},\cdots, l_{\alpha_k} \} \in \{-1, +1 \}^k$ encode whether the corresponding variable is negated or not in the clause. 
The main idea is that each projector $\hat{P}_{\alpha}(\theta)$ then features a 1-eigenvalue subspace where the qubits $\alpha_1,...,\alpha_k$ are all orthogonal to the value favored by the clause, and 
a 0-eigenvalue subspace where the clause $\alpha$ is satisfied. 
For $0< \theta<\pi/2$, the encoding of T 
(true)
and F (false) correspond to non-orthogonal qubit states, hence they are only weakly distinguished. 
This implies that different $\hat{P}_{\alpha}(\theta)$ do not commute in general. 
Nevertheless, the solution space is a common 0-eigenspace of all the $\hat{P}_{\alpha}(\theta)$.
Thus, after mapping a  classical $k$-SAT problem instance onto the $k \cdot m$ parameters $l_{\alpha_i}$, the instantaneous 0-eigenvalue ground subspace of the operator $\hat{O}(\theta)$ at any given $\theta>0$ would then characterize the corresponding $k$-SAT solutions, if there are any. 

In the remainder of this paper, we analyze the state evolution assuming that the $k$-SAT instance does have a solution, implying that $\hat{O}(\theta)$ does feature a simultaneous 0-eigenvalue and we denote $\hat{\Pi}_0(\theta)$ as the projector onto the corresponding subspace. 

In the decision version of $k$-SAT \cite{calabro, Vazirani}, i.e.,~where one asks whether a 0 eigenvalue solution exists or not for a given instance, the above analysis suffices to characterize our ability to provide an arbitrarily accurate answer. 
Indeed, by the definition of NP, any proposed solution can be efficiently verified \cite{Cook_1971, Karp_1972, Levin_1973}. 
When the instance has no solution, any trial will fail. 
When the instance has a solution and we have performed successful Zeno dragging, there is a high probability --- say $\xi$ --- to sample a state that does satisfy the instance; by trying several times, the probability to never observe a valid sample decays exponentially. 
In other words, by declaring ``instance has no solution'' if we have seen none after $\varpi$ trials, we correctly solve the problem with probability $1-(1-\xi)^\varpi$. 
Therefore, the only critical value here is the probability $\xi$ of successful Zeno dragging for $k$-SAT instances that do have a solution. To be useful, $\xi$ should be significantly larger than the probability $D_\mathrm{sol}/2^n$ of uninformed random sampling when the instance has $D_\mathrm{sol}$ solutions; typically, there are hard instances with $D_\mathrm{sol}=O(1)$ \cite{calabro, Vazirani}. 

To initialize the $k$-SAT ``ground state'' for Zeno dragging, we prepare the complete superposition
$|\psi_0\rangle = |+++ ... +\rangle = (\ket{1}+\ket{0})^{\otimes n}/(2^{n/2})$, which is the ground state of any $\hat{O}(\theta_i=0)$. 
From there, the problem instance, the operator $\hat{O}(\theta)$, and hence the solution to the problem instance, are gradually refined as $\theta$ increases (more on this below), until we have aligned our measurement axes with the computational basis at $\theta_f = \pi/2$. 

If the state has been Zeno dragged along the ground subspace of $\hat{O}(\theta)$ to achieve a final accuracy $\xi \in [0,1]$, then the final measurement will provide a $k$-SAT solution with probability $\xi$. This is to be contrasted with the probability $D_\mathrm{sol}/2^n$ that we could obtain by measuring $|\psi_0\rangle$ in the computational basis directly at the start, 
which is just equivalent to guessing possible solutions at random.

There are several non-equivalent ways to implement the Zeno dragging when gradually changing $\theta$ from $\theta_i$ to $\theta_f$ in a general frustration-free setting, since the $\hat{P}_{\alpha}(\theta)$ do not all mutually commute and hence cannot all be projectively measured in parallel. This motivates us to explicitly model the measurement time resolution. 

For the measurement associated to a single clause $\alpha$, we consider Kraus operators $\{ \hat{M}_{r_{\alpha},\alpha}(\theta)\}$ inspired by a ubiquitous setting in quantum optical systems \cite{selective_cqm, CQM.dd, quantum_bayesian, Steinmetz}, where a Gaussian measurement apparatus generates a continuous--valued measurement signal $r_\alpha(t)$ over a measurement duration $\Delta t$:
\begin{equation} \label{eq:Kraus}
    \begin{aligned}
    \hat{M}_{r_{\alpha},\alpha}(\theta)=&\left(\frac{\Delta t}{2 \pi}\right)^{\frac{1}{4}}\bigg\lbrace \exp\left[-\frac{\Delta t}{4}\left(r_{\alpha}+\frac{1}{\sqrt{\tau}}\right)^2\right] \hat{P}_{\alpha}(\theta) +\exp\left[-\frac{\Delta t}{4}\left(r_{\alpha}-\frac{1}{\sqrt{\tau}}\right)^2\right]\left(\hat{\openone} - \hat{P}_{\alpha}(\theta)\right)\bigg\rbrace,
    \end{aligned}
\end{equation}
with $\tau$ the ``characteristic measurement time'' (or inverse of the measurement rate $\Gamma = 1/4\tau$) \cite{RevModPhys.86.1391, PhysRevA.47.642, PhysRevA.50.5256, CQM.dd, BookJordan}, and $\Delta t/\tau$ the measurement strength.
 
One can check that those operators indeed satisfy the fundamental relation $\int dr_{\alpha}\hat{M}^{\dagger}_{r_{\alpha},\alpha}(\theta) \hat{M}_{r_{\alpha},\alpha}(\theta) = \hat{\openone}$ with $r_{\alpha} \in \mathbb{R}$ being the measurement outcome. For a pre-measurement state $\hat{\rho}_0$, the probability of getting outcome $r_{\alpha}$ is $\mathrm{Tr}(\hat{M}_{r_{\alpha},\alpha}(\theta)\hat{\rho}_0  \hat{M}^{\dagger}_{r_{\alpha},\alpha}(\theta))$, and according to Born's rule the post-measurement state $\hat{\rho}_{r_{\alpha}}$ is given as 
\begin{equation}\label{Born_rule}
\hat{\rho}_{r_{\alpha}} = \frac{\hat{M}_{r_{\alpha},\alpha}(\theta) \hat{\rho}_0 \hat{M}_{r_{\alpha},\alpha}^{\dagger}(\theta)}{\mathrm{Tr}(\hat{M}_{r_{\alpha},\alpha} (\theta)\hat{\rho}_0 \hat{M}_{r_{\alpha},\alpha}^{\dagger}(\theta))} \; . 
\end{equation} 
As $\Delta t /\tau \rightarrow \infty$, the measurement defined in this way becomes equivalent to the standard projective measurement of $\hat{P}_{\alpha}(\theta)$; when $\Delta t /\tau \rightarrow 0$, we are in the weak continuous measurement limit where each individual measurement takes an infinitesimal amount of time and perturbs the state infinitesimally; the state dynamics are then described by diffusive stochastic quantum trajectories \cite{Hudson_Parthasarathy, Qsde1, Wiseman1996, Brun2001Teach, Qsde2, gardiner2004quantum, JacobsSteck, barchielli2009quantum, wiseman2009quantum, BookJacobs, BookJordan}.
For intermediate values of $\Delta t /\tau$, the measurement is performed with finite but limited contrast. This measurement model can be naturally realized in superconducting circuit platforms with continuous-time readouts.

For the measurement of several clauses, the Kraus map \eqref{Born_rule} could in principle be generalized to a multi-dimensional continuous signal $\mathbf{r}(t)$, where the state update is conditioned on the readouts of several non-commuting observables. However, it requires some care to correctly model the corresponding Kraus operators, as we effectively move into a situation that is no longer quantum non-demolition (QND) \cite{wiseman2009quantum, barchielli2009quantum, BookJacobs, BookJordan}. 

In the time--continuum limit, issues of measurement commutation may be disregarded, because measurement ordering only affects the dynamics to order $O(\Delta t^2)$ \cite{Jordan_2005, Ruskov_2010, Ruskov_2012, Rouchon_2015, hacohen2016quantum, Ficheux2018, Lewalle02012020, Steinmetz, Karmakar_2022, Jackson_Caves_2023, Wonglakhon_2024}. In other words, in such a system all the clause observables are, simultaneously and infinitesimally weakly, coupled to a separate measurement degree of freedom. Their measurements are thus effectively made in parallel, with backaction terms just adding up in the associated diffusive stochastic quantum trajectories, irrespectively of whether the observables commute or not. 

When measurement signals are obtained after a finite time interval $\Delta t$, we may imagine different operational realities. In one situation, the physical operation would be the same as in the time-continuum limit, but the finite time resolution of the readout only allows us to condition on samples of $\mathbf{r}(t)$ obtained at discrete time steps $\Delta t$.   For such devices, a composite Kraus operator, conditioning backactions on $\mathbf{r}(\Delta t)$, could in principle be assembled and would reflect experimental reality well
\cite{Robinet}. An alternative type of device would, over each small time interval $\Delta t$, measure a (different) set of only commuting observables. In this situation, the Kraus operators \eqref{eq:Kraus} remain a valid model, and their sequential time-ordered application determines the overall evolution. This leads to a more ``digital'' conception of the overall situation, of the type discussed in \secref{sec:Q-circuit} and Fig.~\ref{fig:collision_1Q}. The two distinct physical pictures coalesce in the continuum limit $\Delta t/\tau \rightarrow 0$.

A last point deserving some attention for the multi-channel measurement situation is the normalization of overall measurement rate $\Gamma = 1/4\tau$ (or strength $\Delta t/\tau)$. This part of the device modeling is of course related to the previous discussion about continuous versus sequential measurements. The most optimistic assumption would be that each clause can be measured at a given speed, independent of the fact that it would be measured together with other clauses. 
An intermediate assumption may be that commuting clauses can be measured together, each at a given speed independent of their number, but that measuring, e.g., two non-commuting observables ``together'' means that we must measure each at half its maximal speed (the latter is reminiscent of heterodyne measurement in quantum optics \cite{AK_1965, shapiro_phase_1984, shapiro_quantum_2009, Caves_2012, ochoa2018simultaneous} ). Finally, the most pessimistic assumption would be that the device features a fixed number of output channels, so that the effective measurement speed of each individual clause is proportional to $1/m$, where $m$ is the number of clauses.

In this work, for finite $\Delta t$, we will generally follow a (pessimistic) model compatible with selecting and measuring a single clause at random for every time slice $\Delta t$. In the event that different environmental modes can be used for performing different measurements in parallel on a device, our conclusions will differ at most by the factor $1/m$ on the measurement strength.

\subsection{Unconditional Dynamics}
We begin by considering Zeno dragging based on dissipation; this setting is equivalent to averaging over all possible measurement results, or dissipating information without detecting measurement outcomes at all. 
We note that the QZE works on average, i.e.~autonomously, and this principle underlies the effectiveness of dissipative Zeno dragging.
Performance can be improved by collecting measurement outcomes to perform filtering, post-selection, and/or feedback, but these items are not strictly \emph{necessary} for Zeno based schemes to function.

For any measurement of duration $\Delta t$, averaging over the randomness of both the measurement result $r_{\alpha}$ and also the random selection of clause $\alpha$, the resulting dynamics is then described by the quantum channel: 
\begin{equation}\label{eq:measurement_channel}
\begin{aligned}
    \mathcal{T}(\theta)[\hat{\rho}] &= \frac{1}{m}\sum_{\alpha = 1}^m\int_{-\infty}^{\infty}dr_{\alpha} \hat{M}_{r_{\alpha},\alpha}(\theta) \hat{\rho }\hat{M}^{\dagger}_{r_{\alpha},\alpha}(\theta) \\
    &= \hat{\rho} + \frac{2\beta}{m} \sum_{\alpha=1}^m \left[\hat{P}_{\alpha}(\theta) \hat{\rho} \hat{P}_{\alpha}(\theta) - \frac{1}{2}\left(\hat{P}_{\alpha}(\theta)\hat{\rho} + \hat{\rho} \hat{P}_{\alpha}(\theta) \right) \right],
\end{aligned}
\end{equation}
where $\beta = 1 - e^{-\Delta t/2\tau}$. When $\Delta t /\tau \rightarrow 0$, which is equivalent to the weak continuous limit of the measurement  process, \eqref{eq:measurement_channel} is equivalent to the Lindblad master equation \cite{lindblad1976generators}
\begin{equation}\label{eq:lindblad}
\begin{aligned}
    \frac{d\hat{\rho}}{dt} &= \frac{1}{m\tau} \sum_{\alpha=1}^m \left[\hat{P}_{\alpha}(\theta) \hat{\rho} \hat{P}_{\alpha}(\theta) - \frac{1}{2}\left(\hat{P}_{\alpha}(\theta)\hat{\rho} + \hat{\rho} \hat{P}_{\alpha}(\theta) \right) \right]  = \frac{1}{\tau}\mathcal{L}(\theta) [\hat{\rho}],
\end{aligned}
\end{equation}
where we have defined the Lindbladian superoperator
\begin{equation}
\begin{aligned}
    \mathcal{L}(\theta) [\hat{\bullet}] &= \frac{1}{m}\sum_{\alpha=1}^m \left[\hat{P}_{\alpha}(\theta) \,\hat{\bullet}\, \hat{P}_{\alpha}(\theta) - \frac{1}{2}\left(\hat{P}_{\alpha}(\theta)\,\hat{\bullet} + \hat{\bullet}\,\hat{P}_{\alpha}(\theta) \right) \right] \; .
\end{aligned}
\end{equation}
The characteristic time $\tau$ expresses how fast the output channels convey information about the quantum system. The prefactor $\frac{1}{m}$ again corresponds to the pessimistic assumptions that a single output channel of rate $1/\tau$ has to be shared by all the clauses. 

\subsection{Conditional Dynamics}
The above treatment is based on the average effect of measurement processes, without recording the actual outcomes,  and is thus described by a deterministic quantum channel. In this paper, we will also be interested in the diffusive stochastic dynamics obtained by conditioning on measurement outcomes $\{ r_{\alpha} \}_{\alpha=1}^m$ from simultaneous, weak continuous monitoring of all the $\{ \hat{P}_{\alpha}(\theta)\}_{\alpha=1}^m$. Such dynamics are described by a stochastic master equation, which we here write in Stratonovich form
\cite{GardinerStochastic,bookVanKampen,BookKloedenPlaten,wiseman2009quantum,Lewalle02012020}
\begin{equation}\label{eq:strat_SME}
\begin{aligned}
    d\hat{\rho} 
    &= \frac{1}{m\sqrt{\tau}}\sum_{\alpha = 1}^m \left(r_{\alpha} - \frac{1}{\sqrt{\tau}}\right) \left[\hat{P}_{\alpha}(\theta)\hat{\rho} + \hat{\rho} \hat{P}_{\alpha}(\theta)- 2\hat{\rho}\, \mathrm{Tr}\{\hat{P}_{\alpha}(\theta)\hat{\rho}\}\right] dt\\
    &= \frac{1}{m \sqrt{\tau}} \sum_{\alpha=1}^m \mathcal{F}_{\alpha}(\hat{\rho}, r_{\alpha}, \theta)\, dt\; , \\
    r_\alpha dt &= \frac{2}{\sqrt{\tau}}\mathrm{Tr}(\hat{P}_{\alpha} \hat{\rho})dt + \; dW_{t,\alpha} \; ,
\end{aligned}
\end{equation}
where $dW_{t,\alpha}$ is independent white noise for each $\alpha$ describing the stochastic nature of quantum measurement results, and we have defined $\mathcal{F}_{\alpha}(\hat{\rho}, r_{\alpha}, \theta) =(r_{\alpha} - \frac{1}{\sqrt{\tau}})[\hat{P}_{\alpha}(\theta)\hat{\rho}
 +  \hat{\rho} \hat{P}_{\alpha}(\theta)-  2\hat{\rho} \mathrm{Tr}\{\hat{P}_{\alpha}(\theta)\hat{\rho}\}] $ for later convenience.

\subsection{Alternative Quantum Circuit Implementation \label{sec:Q-circuit}}

\begin{figure}[ptb]
\includegraphics[width = \textwidth]{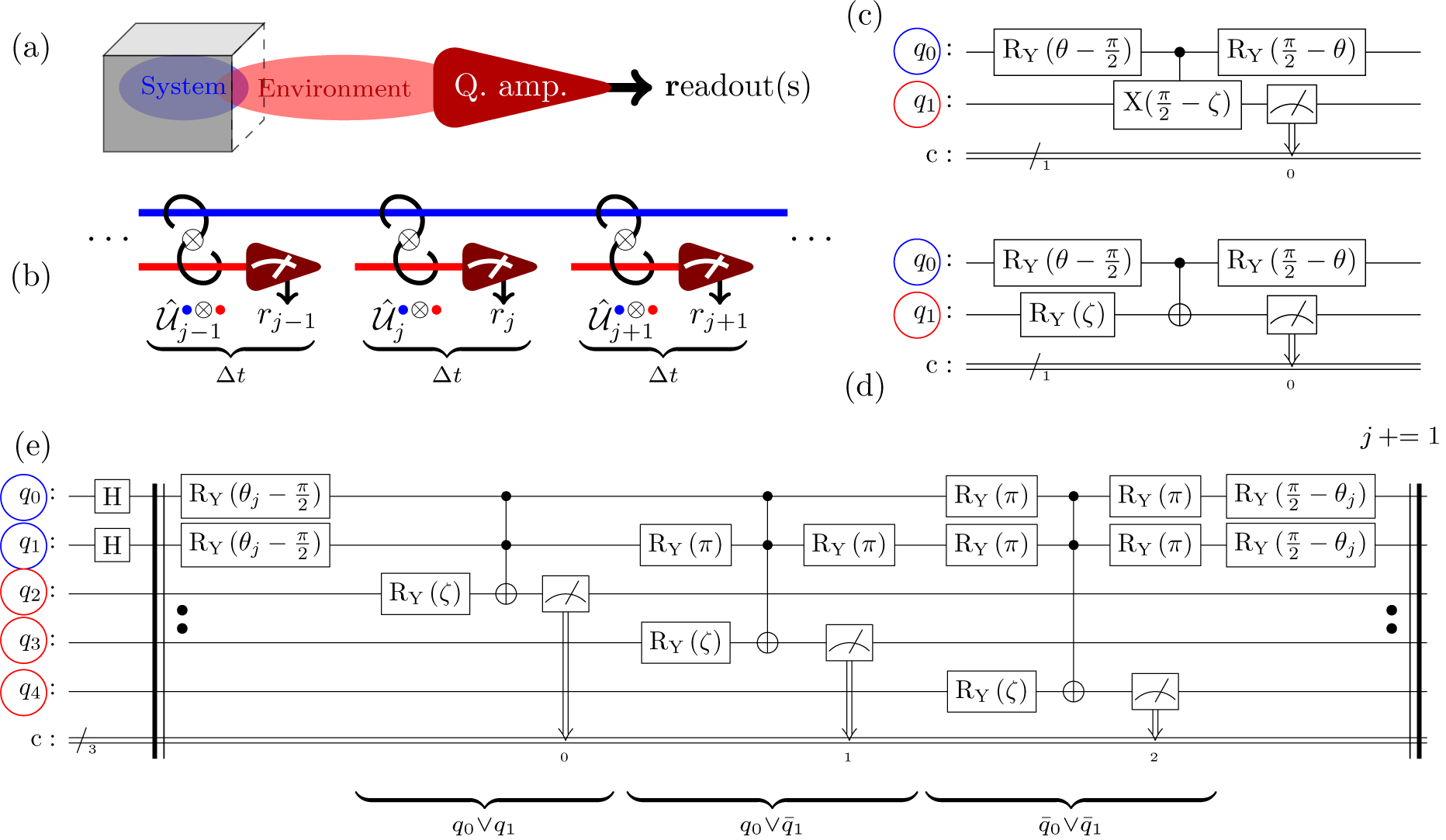}
\caption{
Panels (a) and (b) are adapted from \cite{Philippe_Thesis}, and illustrate a quantum system monitored through its interaction with an auxiliary (environmental) degree of freedom that mediates information exchange with a detector that amplifies some outcomes to a classical scale. 
The representation a) is consistent with e.g.~standard ``collision model'' approaches to modeling open quantum system processes \cite{Ciccarello_2022}. 
Repeated (Markovian) interactions with environmental time--slices, as shown in panel (b), serves as a model for successive generalized quantum measurements, each yielding a readout $r_j$ in the $j^\mathrm{th}$ timeslice. 
The ``environmental'' degrees of freedom mediating information exchange with a detector (red ``meter qubits'') are always reset after they are measured. 
In the limit of infinitesimal $\Delta t$ and infinitesimal correlations generated between the system and environment, we recover continuous weak monitoring from such a picture (or Lindblad dynamics, if we average over the readouts $\mathbf{r}$). 
While panels (a) and (b) apply for any kind of quantum system and environment, as well as for arbitrary strength of any associated measurements, we have drawn panel (b) to be reminiscent of a quantum circuit. 
Panels (c) and (d) make this notion more explicit by expressing models of generalized measurements pertinent to this paper, using only qubits for both the system (blue) and environment (red), assuming $q_1$ is initialized to the computational state $|0\rangle$. 
Specifically, panels (c) and (d) present two options for a quantum circuit that performs a generalized measurement of $q_0$ along the $\hat{\sigma}_\theta = \hat{\sigma}_x\,\cos\theta + \hat{\sigma}_z\,\sin\theta$ axis. 
In this digital setting, the effective measurement strength is set by the degree of entanglement formed between the system $q_0$ and environment $q_1$ and is quantified by the parameter $\zeta$
in panel (c), which shows the implementation of a generalization of the CX gate where $\mathrm{CX}(\tfrac{\pi}{2}-\zeta) = e^{-i\,(\pi/2-\zeta)\,\mathrm{CX}} = \hat{\mathbb{I}}\,\cos(\tfrac{\pi}{2}-\zeta) -i\,\mathrm{CX}\,\sin(\tfrac{\pi}{2}-\zeta)$. This gate interpolates between identity at $\zeta = \pi/2$, and CX at $\zeta = 0$ (up to a phase), and constitutes a direct implementation of the concepts implied in panels (a) and (b) \cite{Brun2001Teach}.  
Panel (d) implements a different version in terms of a CNOT = CX gate, which is conceptually more similar to a no--knowledge measurement \cite{Szigeti_2014, Saiphet_2021}. 
In both panels (c) and (d), projective measurement of $q_1$ in the $\hat{\sigma}_z$ basis effectively implements a projective measurement on $q_0$ with $\zeta = 0$, while for $\zeta \rightarrow \pi/2$ the effective measurement strength vanishes. 
(In panel (c), this is because less entanglement is generated between the system and meter as $\zeta \rightarrow \pi/2$; in panel (d) more correlations may be generated, but they are mis-aligned with the eventual amplification axis, such that only some of the correlations are used, and the remainder are quantum--erased.)
Panel (e) extends this kind of a measurement model to an application of the BZF algorithm using discrete weak measurements for a simple two--qubit $2$-SAT problem \cite{zhang2024_ksat}. The system qubits are prepared in a state $\propto (\ket{e} + \ket{g})^{\otimes 2}$ via Hadamard gates $\hat{\mathrm{H}}$, and then cycles of clause measurements are performed as a simple generalization of the measurement in panel (d), now using CCNOT = CCX = Toffoli gates to generate the appropriate correlations between the system qubits, and the auxiliary qubit that is used to measure each clause. 
}\label{fig:collision_1Q}
\end{figure}

Both versions of Zeno dynamics presented above can be viewed as analog dissipative quantum computation models. 
Alternatively, it is also possible in principle to simulate generalized measurement using the ``quantum digital'' standard circuit model. Here we present a simple circuit implementation of generalized clause measurements with observables corresponding to the projectors \eqref{eq:projector} that are, however, implemented with variable strength in a discrete manner. Such a digital model of generalized measurements~\cite{Brun2001Teach} may be of independent interest to researchers who are more familiar with circuit-model quantum computation. This alternative ``digital'' implementation assumes that qubits function as the ``environment'' that mediates information exchange with the detection apparatus (instead of e.g.~a continuous--variable system such as a resonator); entangling gates between the data and meter qubits create correlations to be leveraged for measurements, leading to discrete measurement outcomes upon projection of the meter qubits.

Example circuits are shown in Fig.~\ref{fig:collision_1Q}. 
For $k$-SAT, each clause measurement requires some single qubit rotations and a single $k$-controlled gate. If there are $n$ variables with $m$ clauses, we then need at most $n + m$ qubits to implement all clause measurements, where $n$ are data qubits requiring long coherence times and memory, while the $m$ auxiliary qubits can be less good because they are being used to mediate readout, and are reset every cycle. Strictly speaking, a single auxiliary qubit may be recycled for all measurements, if it can be reset extremely fast. The convenience of the analog versus digital implementations depend on the constraints of a given experimental setting. In particular, we stress that from the perspective of the models presented above, and on which the results of the present manuscript are based, the main difference between the ``analog'' setting and the gate--based one of Fig.~\ref{fig:collision_1Q}, is that the analog situation is natural to systems with time--continuous dissipation (i.e.~measurements mediated by always--on dissipation), whereas the circuit implementation is more amenable to situations where measurements are implemented discretely and in a specific order. Between the variations spanned by these choices, we believe that (noisy) implementations of small--scale Zeno--Dragging $k$-SAT problems are feasible in the near term on a diversity of experimental platforms.

This paper is focused on the analog implementation, motivated by the fact that continuous-time readout is the natural setting in superconducting circuit platforms (for specific applications to qubits, see e.g., \cite{hacohen2016quantum, murch2013observing, campagne2016observing, campagne2016using, PhysRevLett.112.170501, PhysRevX.7.031023}). The above digital implementation is presented here as an existence proof, without any claim to be particularly efficient in practice.

\subsection{Summary of the setting}

In the next section, we will analyze the dynamics of Zeno dragging using the general expression of the unconditional measurement process \eqref{eq:measurement_channel}, or in the continuous limit using the unconditional Lindblad dynamics \eqref{eq:lindblad}. 
This provides the average performance in the absence of measurement-based feedback. 
One restriction in the present analysis is that we assume the same $\theta$ and the same $1/\tau$ at all times for all components $\alpha$. This may  be generalized in future work if scheduling different qubits or different clauses individually suggests significant advantages.

\subsection{Summary of notation}

Throughout the paper we denote $\hat{O}(\theta) = \frac{1}{m}\sum_{\alpha=1}^m\hat{P}_{\alpha}(\theta)$ as the frustration-free Hermitian operator that we would like to optimize, where $\hat{P}_\alpha(\theta)$ are $k$-local clause projectors,  $\theta$ is the control parameter, and $m$ is the number of clause projectors.  We denote $\theta_i$ and $\theta_f$ as the initial and final values of $\theta$, and use $\alpha$ to index the clauses. In the measurement setting, we denote $\hat{M}_{r_\alpha, \alpha}(\theta)$ as the Kraus operators, where $r_\alpha$ is the measurement outcome corresponding to clause $\alpha$. We denote $\tau$ as the characteristic measurement time, $\Gamma$ as the measurement rate and $\Delta t$ as the duration for a single measurement. We use $\beta = 1 - e^{-\Delta t/2\tau}$ {as an indicator of measurement strength, where $\beta = 1$ corresponds to projective measurement and the limit $\beta \rightarrow 0$ corresponds to continuous weak measurement}. In the weak continuous limit, the generator of the unconditional dynamics is denoted as $\mathcal{L}(\theta)$, while $\mathcal{F}(\theta)$ plays a similar role for conditional dynamics, with $dW$ denoting the Wiener noise associated with the diffusive continuous measurements. We denote the measurement channel as $\mathcal{T}(\theta)$.

We assume that $\hat{O}(\theta)$ {possesses} a nontrivial {ground-state space} associated with eigenvalue 0. We use $\hat{\Pi}_0(\theta)$,  a  function of the control parameter $\theta$, to denote the projector onto this {ground-state space}. $\mathrm{Tr}[\hat{\Pi}_0(\theta)]$ is the dimension of the {ground-state space}. We will use $G(\theta)$ {to denote} the spectral gap of $\hat{O}(\theta)$ and  $f(\theta)$ to denote the fidelity with respect to $\hat{\Pi}_0(\theta)$. Finally, we denote the number of qubits as $n$, the number of $\theta$-increments as $N$, and the number of times the measurement channel $\mathcal{T}(\theta)$ is applied at a fixed value of $\theta$ as $M(\theta)$. {We shall often omit the explicit $\theta$ dependence of $\hat{\Pi}_0(\theta)$, $G(\theta)$, $f(\theta)$, $\mathcal{T}(\theta)$, and $M(\theta)$, for notational convenience, but note that they should all be understood as $\theta$-dependent unless otherwise specified.}

In {the} setting of  optimal control, we denote $J$ as the cost function, and $\mathcal{S}$ as the stochastic action. We use $\mathbf{q}$ and $\mathbf{\Lambda}$ to denote the coordinates and costates of the dynamical systems being optimized, $\mathbf{u}$ to denote their controls and $\mathbf{f}$ to denote their dynamical functions. The Pontryagin Hamiltonian will be denoted as $\mathbb{H}$.

\section{Convergence of Multi-channel Zeno Dragging}\label{sec_convergence}

In this section, we provide conditions under which the convergence towards a state with significant overlap on the target states is guaranteed in multi-channel Zeno dragging. 
The result can be viewed as a version of the adiabatic theorem adapted to open quantum systems driven by the Lindblad equation \cite{Sarandy_2005,Vacanti_2014,Campos_2016}, instead of Schr\"odinger equation in closed systems. Similarly to the latter and more standard adiabatic theorem \cite{RevModPhys_adiabatic}, our convergence result requires modulation of $\theta(t)$ slowly enough, compared to the inverse of a spectral gap $G(\theta)$, defined to be $G(\theta) = \min_{i}\{\lambda_i(\theta) >0\}$, where the $\lambda_i(\theta)$ are eigenvalues of the target Hermitian observable $\hat{O}(\theta)$. Perhaps less intuitively, this spectral gap is not the one of the Lindbladian $\mathcal{L}(\theta)$, but instead the gap of the target Hermitian observable $\hat{O}(\theta)$, which may be much larger. We will also see in this section how this theoretical condition captures the information loss incurred by our model at finite $\Delta t$. Indeed, as already empirically observed in \cite{zhang2024_ksat}, the weak continuous limit $\Delta t /\tau \rightarrow 0$ will be optimal for our criterion too, although this only concerns pre-factors, i.e.~the scaling of the algorithm for large $n$ and $m$ remains the same.

\subsection{Convergence by Lindbladian Mixing}
One immediate  observation of the dynamics described by \eqref{eq:measurement_channel} is that, for any fixed $\theta$, it keeps the state fidelity to the $\theta$-associated ground-state space unchanged, i.e.
\begin{equation}
    \mathrm{Tr}\left\{\mathcal{T}(\theta)[\hat{\rho}] \; \hat{\Pi}_0(\theta)\right\} = \mathrm{Tr}\left\{\hat{\rho}\; \hat{\Pi}_0(\theta)\right\} \; ,
\end{equation}
which is proved in Appendix \ref{appendix:Lower bound of the solution state fidelity} Lemma \ref{lemma_decomposition}. At first glance, this seems suspicious, because it states that quantum Zeno dragging would never increase fidelity to the ground-state space: when the measurement basis $\theta$ is changed, fidelity typically decreases, while under Kraus map / Lindblad evolution, it just stays constant. However, like in adiabatic computing, the point is to start in a particular situation with high fidelity (typically $\mathrm{Tr}\{\hat{\rho}(0) \hat{\Pi}_0(\theta_i)\} = 1$), and then ensure that the total drop in fidelity remains limited before reaching $\theta_f$. 
In Proposition \ref{proposition_mixing} below, we show that this is precisely what the Kraus map / Lindblad evolution ensures, as illustrated on Fig.~\ref{fig:theorem_visualization}: once the Kraus map/Lindblad evolution has reduced the coherence between the ground-state space and the other eigenspaces of $\hat{O}(\theta)$, the fidelity of the resulting mixed state to the ground-state space of $\hat{O}(\theta+\Delta\theta)$ is much higher than for a general state featuring the same fidelity to the ground-state space of $\hat{O}(\theta)$. We refer to this mechanism as Lindbladian mixing. 
In other words: simply moving $\theta$ from $\theta_i$ to $\theta_f$ in small increments $\Delta\theta$ without ever applying \eqref{eq:measurement_channel} --- thus without ever changing the state --- would in general provide an excessively low overlap between the final state (which now equals the initial state) and $\Pi_0(\theta_f)$, and this overlap is of course independent of the chosen schedule $\theta(t)$; however, it is sufficient to interleave the small $\Delta\theta$ increments with dynamics inducing decoherence between ground-state space and non-ground-state space of $\hat{O}(\theta)$, in order to allow dragging the state arbitrarily close to 100\% overlap with $\Pi_0(\theta_f)$. 
This not only elucidates the mechanism by which quantum Zeno Dragging works, but it also provides a guarantee on success probability, equivalent to an upper bound on expected time-to-solution, as stated in the main result Theorem \ref{theorem_convergence} below. The full proofs of these results can be found in Appendix \ref{appendix:Lower bound of the solution state fidelity} and Appendix \ref{appendix:convergence of zeno} respectively.

\begin{figure}[tbhp]
    \centering
    \begin{tikzpicture}
        \node at (-6.2, 0) {\includegraphics[width=0.327\linewidth, trim=2cm 1.4cm 0cm 1.4cm, clip]{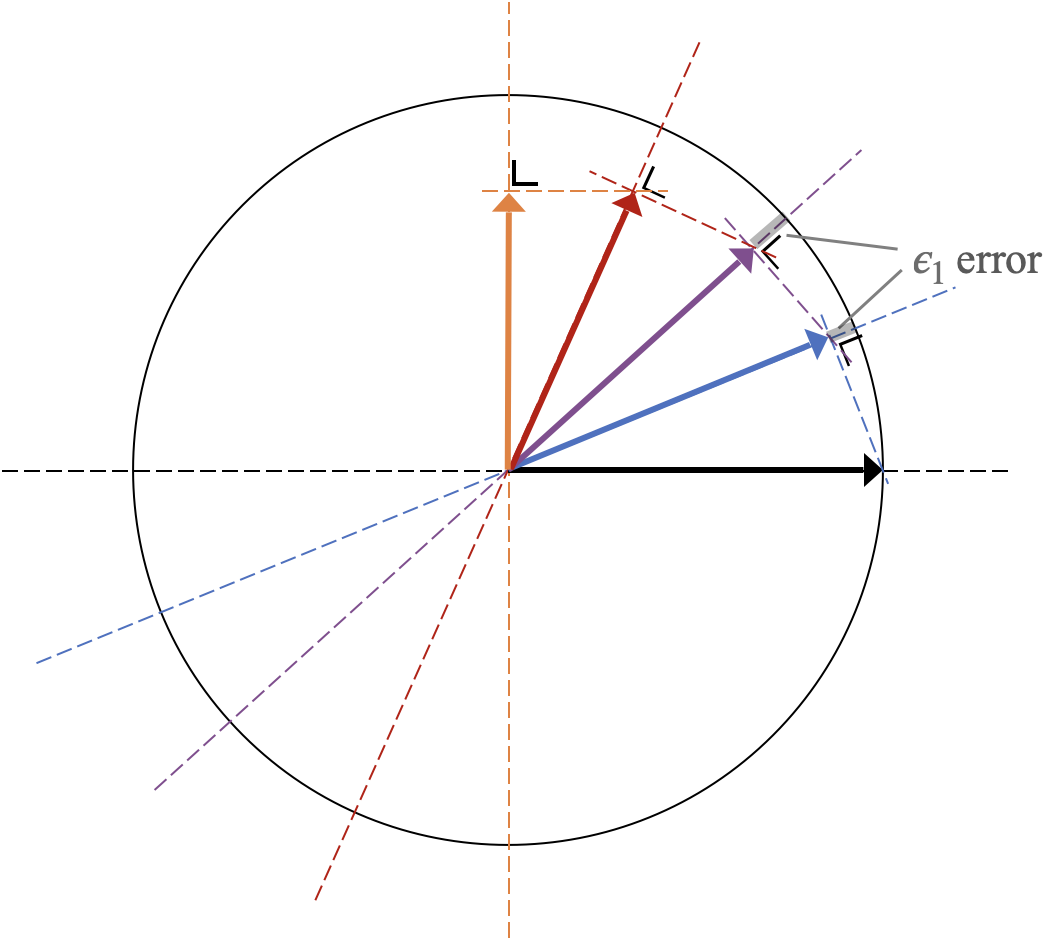}};
        \node at (-6.2-0.3, -2.5) {(a)};
        \node at (-1, 0) {\includegraphics[width=0.327\linewidth, trim=2cm 1.4cm 0cm 1.4cm, clip]{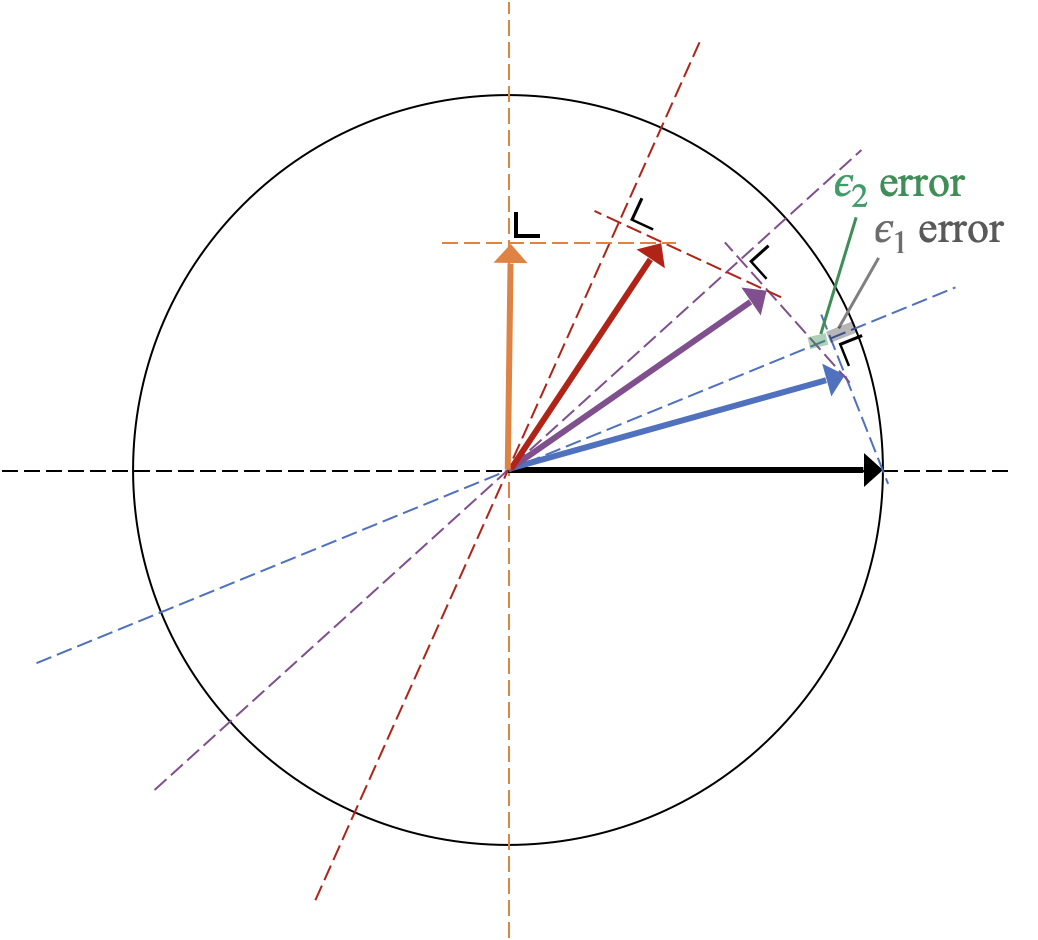}};
        \node at (-1-0.3, -2.5) {(b)};
        \node at (4.2, 0) {\includegraphics[width=0.327\linewidth, trim=2cm 1.4cm 0cm 1.4cm, clip]{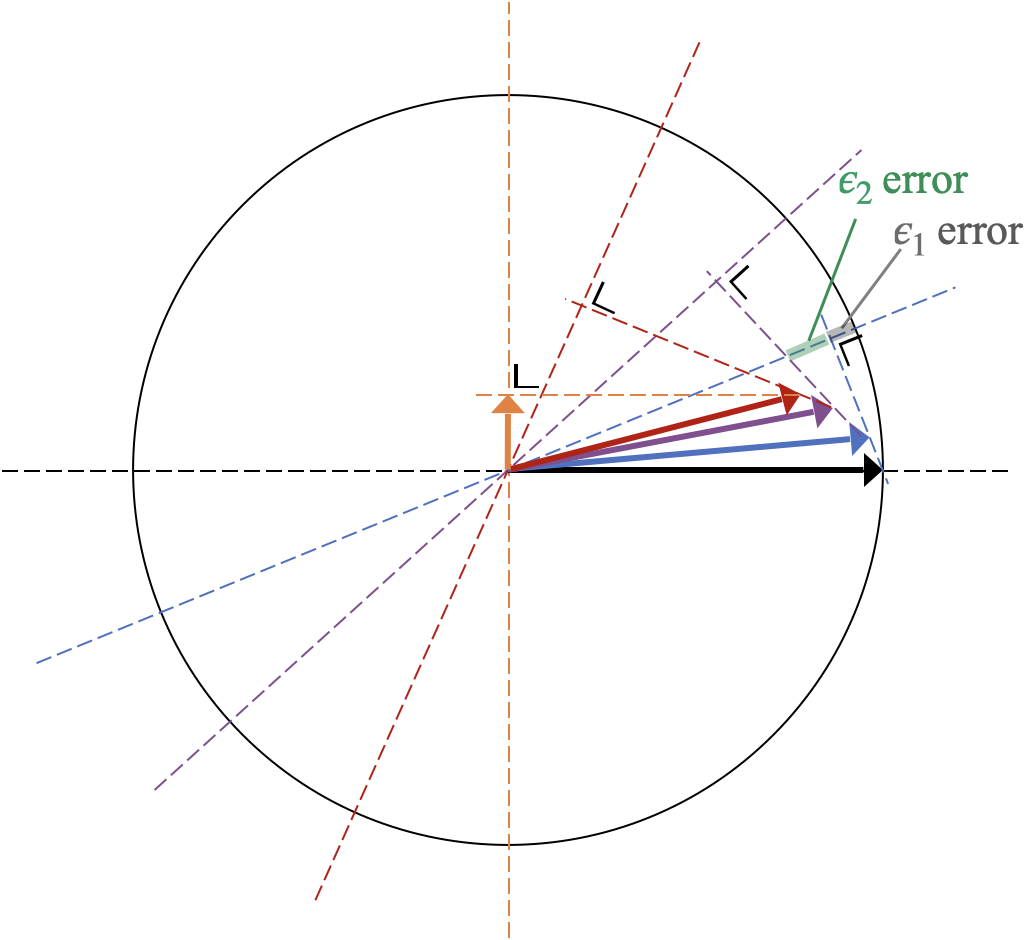}};
        \node at (4.2-0.3, -2.5) {(c)};
    \end{tikzpicture}
    \caption{Schematic mechanism of Zeno dragging on a qubit. The circle represents the Bloch sphere cut by the ZX plane. An initial state $\ket{+}$ (black arrow) undergoes $N=4$ dragging steps (various colors) with the goal to reach the $\ket{1}$ state (upwards vertical). In the unconditional evolution, each step induces decoherence of the current state (colored arrow) perpendicularly to the current measurement axis (colored dashed line). When this decay is complete before each update of the measurement axis (panel (a)), the corresponding loss in purity accumulates in what Theorem \ref{theorem_convergence} calls $\epsilon_1$ type errors only. When the measurement-induced decay is only partially completed, an additional $\epsilon_2$ type error builds up (panels (b) and (c)), with drastic consequences when the measurement-induced decoherence only shortly drags the state at each step (panel (c)): a final measurement along the vertical axis then yields fidelity to $\ket{1}$ barely above 50\%.}\label{fig:theorem_visualization}
\end{figure}
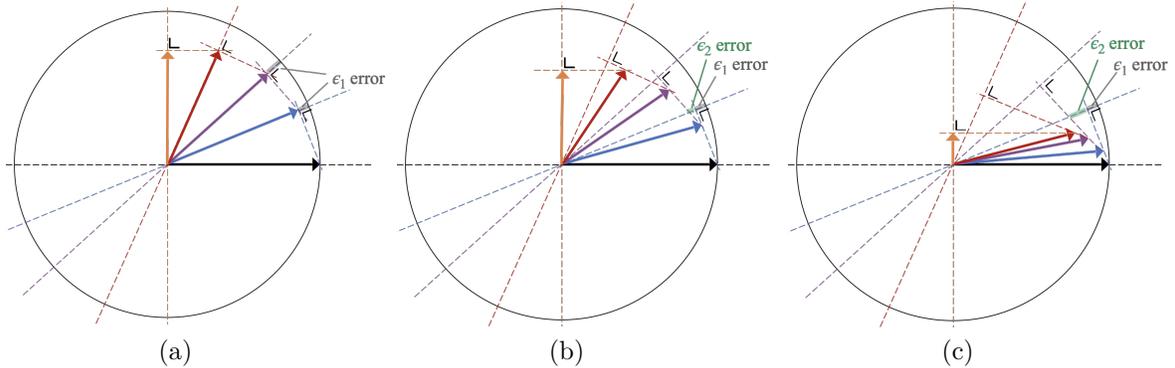

\begin{restatable}[Lower bound on the fidelity]{proposition}{mainproposition}\label{proposition_mixing}
    Let the projectors onto the 0-eigenvalue ground-state space of $\hat{O}(\theta)$ at some $\theta$ and $\theta +\Delta \theta$ be $\hat{\Pi}_0(\theta)$ and $\hat{\Pi}_0(\theta + \Delta \theta)$. Characterize the overlap between these two spaces by $\langle \psi_0 | \hat{\Pi}_0(\theta + \Delta \theta) | \psi_0 \rangle \geq \delta >0$ for all normalized $|\psi_0 \rangle$ satisfying $\hat{\Pi}_0(\theta) \, |\psi_0 \rangle = |\psi_0 \rangle$. If the quantum state $\hat{\rho}$ has fidelity $f$ with respect to $\hat{\Pi}_0(\theta)$, then after application of $\mathcal{T}(\theta)$, the fidelity $f^{\prime}$ with respect to $\hat{\Pi}_0(\theta + \Delta \theta)$ is lower bounded by    
    \begin{equation}\label{eq:fidelity_drop}
        f^{\prime} \geq f \delta - 2\left[1 - \beta \, G(\theta) \right]\,\sqrt{a(\delta)\,f(1-f)} \; ,
    \end{equation}
where $G(\theta)$ is the spectral gap between 0 and the second lowest eigenvalue of the operator $\hat{O}(\theta)$, and $a(\delta)=\delta-\delta^2$ for $\delta \in [1/2,1]$ or $a(\delta)=1/4$ for $\delta \in [0,1/2]$.
\end{restatable}

When $\mathcal{T}(\theta)$ is applied $M$ times for a fixed $\theta$, the above \eqref{eq:fidelity_drop} is consequently modified to be
\begin{equation}\label{eq:fidelity_drop2}
        f^{\prime} \geq f \delta - 2\left[1 - \beta \, G(\theta) \right]^M\,\sqrt{a(\delta)\,f(1-f)} \, .
\end{equation}
This means that the loss of fidelity can be made exponentially close to $f(1-\delta)$, by applying the measurement channel multiple times. By writing $\delta= \cos(\Delta\phi)$ for some small $\Delta\phi$, we see that reaching this $f(1-\delta)$ regime essentially means that we can move $\hat{\Pi}_0(\theta)$ by an angle $\Delta\phi$ while only losing fidelity proportional to $(\Delta\phi)^2$. When covering a finite angle $\phi$ with $\phi/\Delta\phi$ such steps, the loss in fidelity converges to zero as $\Delta\phi \rightarrow 0$. The number of iterations $M$ needed at each $\theta$ value to maintain this regime is inversely proportional to the spectral gap $G(\theta)$, indicating that the schedule of $\theta(t)$ is important in order to obtain optimized performance.

\begin{restatable}[Convergence of BZF Zeno dragging; see Fig.~\ref{fig:theorem_visualization}]{theorem}{maintheorem}\label{theorem_convergence}
    Let $\Delta t$, $\epsilon_1$, $\epsilon_2 > 0$. Assume\footnote{This precise assumption is inspired by the example of $k$-SAT on $n$ qubits according to \eqref{eq:projector}, \eqref{eq:projector2}, where it obviously holds. It allows in general for a dependence of the ground-spaces overlap on problem size $n$.} $$
    \langle \psi_0 | \hat{\Pi}_0(\theta + \Delta \theta)|\psi_0\rangle\geq \left(\mathrm{cos}(\tfrac{\Delta \theta}{2})\right)^{2n} $$
    for any normalized $|\psi_0\rangle$ in the solution space at $\theta$ i.e.~for any $|\psi_0\rangle$ satisfying $\hat{\Pi}_0(\theta) |\psi_0\rangle = |\psi_0\rangle$. Suppose that one performs Zeno dragging for this system with initial solution space fidelity $1$ at $\theta_i$ and ends at $\theta_f > \theta_i$ using a linear schedule, i.e.~constant increments $\Delta \theta$. Then to ensure a final solution space fidelity $f_N \geq 1 - \epsilon_1 - \epsilon_2$, it is sufficient to take a total number of $\theta-$increments $N = \frac{\theta_f-\theta_i}{\Delta \theta}\geq \frac{n(\theta_f-\theta_i)^2}{4\epsilon_1}$, with each step including $M(\theta)$ applications of the channel $\mathcal{T}(\theta)$, where 
    \begin{equation*} \begin{split}
    M(\theta) &\geq \frac{\mathrm{log}(\frac{1}{\epsilon_2}) + \frac{1}{2}\mathrm{log}(n) + \mathrm{log}(\theta_f-\theta_i)}{\mathrm{log}(\frac{1}{1-\beta  G(\theta)})} \approx \frac{\mathrm{log}(\frac{1}{\epsilon_2}) + \frac{1}{2}\mathrm{log}(n) + \mathrm{log}(\theta_f-\theta_i)}{\beta G(\theta)} \; .
    \end{split} \end{equation*}     
\end{restatable}    
Obviously, the same holds if $M(\theta)$ is based on a lower bound $G_{\min}$ on $G(\theta)$. This bound would be less tight but still feature the same bottlenecks. For $\beta = 1$, this recovers the limit of projective measurement. The main idea in the above theorem is that the decrease of fidelity comes from two sources: the population flipping due to the change of the measurement axis ($\epsilon_1$ term), assuming the state was block-diagonal in zero-eigenspace versus nonzero-eigenspace of $\hat{O}(\theta)$; and the error accumulated due to incomplete decay of the coherences between zero-eigenspace and nonzero-eigenspace of $\hat{O}(\theta)$ at each $\theta$ value ($\epsilon_2$ term). 
The role of the spectral gap only comes into play in the second type of error. The result follows from bounding both $\epsilon_1$ and $\epsilon_2$ types of error. First, choosing $\Delta\theta\leq \tfrac{4 \epsilon_1}{n(\theta_f-\theta_i)}$ maintains a sufficiently low $\epsilon_1$ type error, i.e. a low overall fidelity decrease, if we assume that the solution and non-solution eigenspaces of $\hat{O}(\theta)$ have been well mixed at each step ($f(1-\delta)$ regime of fidelity loss in Proposition \ref{proposition_mixing}); second, computing $M$  according to \eqref{eq:fidelity_drop2} bounds the cumulative $\epsilon_2$ error arising from discarding, at each increment of $\Delta\theta$, the residual coherences between the zero-eigenspace and the nonzero-eigenspace of $\hat{O}(\theta)$. 
Here, well-mixed means that coherences between
the solution and non-solution spaces have been forced to decay.

When the number of qubits $n$ becomes large, the number of $\theta$-increment steps required increases linearly. 
This works thanks to the fact that, after perfect mixing, the useful part of the state upon applying a $\Delta\theta$ increment decreases only by a factor $\left(\mathrm{cos}(\frac{\Delta \theta}{2})\right)^{2n}$. 
With $\Delta\theta^2$ scaling as $O(1/n^2)$, the multiplication of $N=O(n)$ such factors remains bounded away from zero. 
The bottleneck in the convergence speed would then likely come from ensuring good mixing at each step, i.e.~the scaling of $G(\theta)$ with the number of qubits $n$. 
In the $k$-SAT problem, this gap increases with $\theta$ and is thus lowest for $\theta_i$ close to $0$. 
In this regard, it may be interesting to note that a first step of size $\Delta\theta_0 = O(1/\sqrt{n})$ is also acceptable in Theorem \ref{theorem_convergence}: it implies a larger error, but only once instead of $n$ times, leading to the same overall result. 

The dependence of the convergence on the spectral gap implies the importance of optimizing the schedule, i.e., the evolution of $\theta(t)$ over time, which is studied in Sections \ref{sec_OCT} and \ref{sec_numerics}. This feature is analogous to the situation in
adiabatic quantum computing, and Theorem \ref{theorem_convergence} can thus also be considered as an adiabatic theorem for open quantum systems under continuous monitoring. However, there are some important differences, as we discuss  further in Section \ref{sec:3.3} below.

\subsection{Advantage of Weak Continuous Limit}

When considering the total dragging time needed to achieve a desired accuracy, the criterion of Theorem \ref{theorem_convergence} favors the limit $\Delta t/\tau \rightarrow 0$, as stated in the following Corollary. This captures an unavoidable loss of information when claiming to measure several non-commuting observables after a finite time interval $\Delta t$ and without looking at the details of how the signal builds up during the time interval. It also matches earlier empirical observations for quantum Zeno dragging \cite{zhang2024_ksat}.

\begin{restatable} [advantage of weak continuous limit]{corollary}{maincorollary}\label{corollary_weak_advantage}
    The total time needed to guarantee a final solution state fidelity $f \geq 1 - \epsilon_1 - \epsilon_2$, according to the criterion of Theorem \ref{theorem_convergence} and assuming a constant measurement setting ($M$,$\Delta t$) for all $\theta$, is: 
    \begin{equation*} 
    T \geq \frac{\mathrm{log}(\frac{1}{\epsilon_2}) + \frac{1}{2}\mathrm{log}(n) + \mathrm{log}(\phi)}{G_{min}} \cdot \frac{n\phi^2}{4\epsilon_1} \cdot \Upsilon \quad\text{with}\quad \Upsilon \equiv  \frac{\Delta t}{1-e^{-\Delta t/2\tau}}.
    \end{equation*}
\end{restatable} 
The factor $\Upsilon$ is a monotonically increasing function of $\Delta t$. Therefore, the total dragging time attains its optimal value at $\Delta t = 0$, where $\Upsilon = 2\tau$. The projective limit in contrast corresponds to $\Delta t = s \, 2 \tau$ with $s \gg 1$ and thus $\Upsilon \simeq \Delta t$ becomes $s$ times larger, as does the total time $T$. While $s$ can be a large factor, the advantage brought by weak continuous measurement still only reduces $T$ by this constant: it does not affect how $T$ scales with $n$ (both directly, and indirectly through $G_{\min}$) and other parameters. This is also consistent with numerical observations in \cite{zhang2024_ksat}. Note that this advantage is based on the inherently continuous-time model of Zeno dragging, and may not necessarily hold for the ``digital'' version of Fig.~\ref{fig:collision_1Q}, in the event that the operations described there each require additional implementation time overheads that are absent in the present analysis. For more detailed discussion about the role of the measurement strength $\Delta t/\tau = 4\,\Delta t\,\Gamma$ that characterizes the Kraus map and the origin of the advantage of weak continuous limit, particularly in the case of non-commuting measurements, see Appendix \ref{appendix:convergence of zeno}.

Before concluding this section, we wish to emphasize the compatibility of the advantage of the weak continuous limit with existing quantum hardware. The advantage in convergence time attained in this limit is only possible in quantum hardware that can function in analog mode, such as superconducting quantum hardware for which the natural qubit measurement modality is a variable strength continuously monitored measurement. Such platforms have been extensively used to continuously monitor quantum trajectories of qubits \cite{murch2013observing, PhysRevLett.112.170501,campagne2016observing,campagne2016using, PhysRevX.7.031023,hacohen2016quantum}. In contrast, in other quantum hardware  platforms that are restricted to discrete-time logical gate operations such as the circuits in panels c) - d) of Fig. \ref{fig:collision_1Q}, the number of gates will necessarily diverge in the weak-continuous limit.   

In the rest of this work, we shall focus on the weak continuous limit $\Delta t /\tau \rightarrow 0$. Also, for notational convenience, we set $\tau =1$ from now on, such that any time scale discussed hereafter is in units of $\tau$.

\subsection{Summary of Convergence Results}\label{sec:3.3}

With Theorem \ref{theorem_convergence} and its proof, we have shown here that the convergence of the Zeno-dragging dynamics towards the ground state of $\hat{O}$ is guaranteed when two conditions are met. These are that: 1) the control parameter $\theta$ moves slowly, and 2) the number of times the measurement channel $\mathcal{T}$ is applied, $M$, is large compared to $1/G(\theta)$, where $G(\theta)$ is the spectral gap of the Hermitian  operator $\hat{O}(\theta)$. The first condition is easy to satisfy, while the second is usually the bottleneck in performance of the algorithm. We note that $\hat{O}(\theta)$ is a sum of projectors and is not the system Hamiltonian, making this second condition quite different from the analogous condition for adiabatic quantum evolution. As a corollary, we showed that the weak continuous limit $\Delta t/\tau \rightarrow 0$ gives an optimal prefactor in the bound of the convergence time, which provides theoretical support for our previous numerical evidence for the superiority of weak continuous limit in \cite{zhang2024_ksat}.

We note that Theorem \ref{theorem_convergence} shows some similarity to the  adiabatic theorem for quantum annealing and could possibly be considered as an adiabatic theorem for open quantum systems. However, there are several differences between the result of Theorem \ref{theorem_convergence} and the adiabatic theorem for quantum annealing and for adiabatic quantum computing. Firstly, in adiabatic quantum computing and quantum annealing, one usually does not know the location of the minimum gap along the path of the control parameter.  This makes it very hard to optimize the adiabatic or annealing schedule. \cite{Roland_Cerf} provides an exception to this. In contrast, as we show in Section \ref{sec_numerics}, for the measurement-driven algorithm the spectral gap of $\hat{O}(\theta)$ is a monotonically increasing function of $\theta$ and we know that the minimum gap is at the starting point. This makes it significantly easier to obtain a schedule that adapts to the variation of the spectral gap with $\theta$,  which can provide a significant advantage. 
This was first noted for adiabatic quantum computing in \cite{Roland_Cerf}, where the quadratic speedup provided by Grover's algorithm for unstructured search was shown to be recovered by adiabatic quantum algorithms, provided 
that one uses a schedule that follows the spectral gap locally. 
An additional key difference from quantum annealing is that our proposed approach allows the dynamics to be heralded and thereby enables additional operations conditioned on measurement outcomes. This can further improve the performance of the protocol in terms of target state fidelity.  We will provide some results in this direction in Section \ref{sec_numerics}, and view this as an interesting opening to further improve the performance of Zeno dragging in future work. Finally, in terms of the convergence time $T$ scaling with the spectral gap $G$, we note that the quantum annealing literature usually gives $T=\mathcal{O}(\Vert \partial H \Vert/G^2)$ with a strict lower bound $\geq \mathcal{O}(1/G)$ \cite{RevModPhys_adiabatic}. For frustration-free Hamiltonians, as we have here, the spectral gap can be amplified by a square root factor \cite{somma2013spectral}, allowing quantum annealing to explicitly attain a performance in $T = \mathcal{O}(\Vert \partial H \Vert/G)$. 
Our result in Theorem \ref{theorem_convergence} matches this better lower bound in the k-SAT setting, which also improves upon the early results for measurement-driven algorithms in \cite{BZF_PRA, PhysRevA.66.032314}. 
\\

\noindent \textbf{Remark:} Since Theorem \ref{theorem_convergence} expresses the convergence time with respect to the spectral gap of a frustration-free Hamiltonian, one might ask whether spectral gap amplification techniques \cite{somma2013spectral} could lead to a further quadratic speedup. Applying spectral gap amplification here is not trivial, since it leads to a Hamiltonian whose components are no longer bare projectors and whose zero-energy subspace contains the solution, which however sits now in the middle of the spectrum rather than at the minimum of this. As a consequence, the resulting modified Hamiltonian cannot be easily mapped back to Zeno measurement operators. 

\section{Scheduling Optimization With Optimal Control Theory }\label{sec_OCT}
The above analysis provides a theoretical guarantee of Zeno dragging convergence, and shows the optimality of weak continuous limit. However, it does not directly tell us the optimal schedule $\theta(t)$ for performing Zeno dragging. The question is not trivial: should we rather quickly jump over ``bad'' zones, or slow down there to avoid losing our state? In this section, we take a perspective on this problem from optimal control theory and show how to derive the optimal schedule under general circumstances with various demands for optimality. 
In particular, for both unconditional dynamics \eqref{eq:lindblad} and conditional dynamics \eqref{eq:strat_SME}, we will see two types of optimal control tasks naturally arise: the optimal final state problem and the optimal time problem \cite{kirk2004optimal, d2021introduction, quantum_control_review1}. We will outline the optimality conditions and also the procedure for performing optimization for both versions of the problem. Section \ref{sec:4.1} formulates and solves the optimal control problem for unconditional, Lindblad dynamics, deriving first a measurement schedule ensuring optimal fidelity with the solution state at a given final time in Section \ref{OC_costate}, and then formulating and solving for the measurement schedule that minimizes the time-to-solution (TTS) in Section \ref{sec:opt_time}. In Section \ref{sec:4.2} we then address the optimal control problem for conditional dynamics, using a most likely path approach. Section \ref{sec:4.2.1} derives the corresponding measurement schedule that optimizes the fidelity with the solution state over a given time interval, and Section \ref{sec:4.2.2} derives the corresponding measurement schedule that minimizes the TTS.

Before deriving the optimal measurement schedules for these four different settings, we first summarize the optimal control approach underlying our derivations. This is the Pontryagin Maximum Principle (PMP), 
which relates optimal control of time--continuous dynamics to action functionals and variational calculus. Suppose there are some state coordinates $\mathbf{q}(t)$, the system dynamics $\dot{\mathbf{q}}=\mathbf{f}(\mathbf{q},\mathbf{u})$ driven by controls $\mathbf{u}(t)$, the initial state $\mathbf{q}(T_i)=\mathbf{q}_i$, and some total cost function to be minimized
\begin{equation}
    J = J_{T_f}(\mathbf{q}(T_f),  T_f) + \int\limits_{T_i}^{T_f} dtJ_t(\mathbf{q}(t), \mathbf{u}(t)),
\end{equation}
where the terminal cost $J_{T_f}$ is a cost function at the final time of the system's controlled evolution, and the running cost $J_t$ is a dynamic contribution to the cost. One may always define an action 
\cite{pontryagin2018mathematical, PMP_tutorial, DAllessandroBook, Book_OCT_SL}
\begin{subequations}\label{PMP_traditional}\be 
\mathcal{S} = J_{T_f} + \int\limits_{T_i}^{T_f}dt\left(\boldsymbol{\Lambda}\cdot\dot{\mathbf{q}} - \boldsymbol{\Lambda}\cdot\mathbf{f} + J_t \right) = 
J_{T_f} + \int\limits_{T_i}^{T_f}dt\left(\boldsymbol{\Lambda}\cdot\dot{\mathbf{q}} - \mathbb{H} \right) \ee
\be 
\text{for}\quad \mathbb{H}(\mathbf{q},\mathbf{u},\boldsymbol{\Lambda}) \equiv \boldsymbol{\Lambda}\cdot\mathbf{f}(\mathbf{q},\mathbf{u}) - J_t(\mathbf{q},\mathbf{u}),\label{H_vec} 
\ee \end{subequations}
where the scalar $\mathbb{H}$ is the Pontryagin Hamiltonian. The costate vector $\boldsymbol{\Lambda}(t)$ here acts both as ``momenta'' conjugate to the state coordinates $\mathbf{q}$ in the usual sense of Hamiltonian mechanics, and as Lagrange multipliers constraining the system's evolution to physically--motivated equations of motion $\dot{\mathbf{q}} = \mathbf{f}$ \cite{PMP_tutorial, DAllessandroBook, Book_OCT_SL}.
The PMP  states that optimal controls $\mathbf{u}^\star(t)$ satisfy a variational principle $\delta\mathcal{S} = 0$, which yields
\be 
\dot{\mathbf{q}} = \partl{\mathbb{H}}{\boldsymbol{\Lambda}}{},
\quad 
\dot{\boldsymbol{\Lambda}} = -\partl{\mathbb{H}}{\mathbf{q}}{}, \quad
\mathbf{u}^\star(t) = \mathop{\arg\max}\limits_{\mathbf{u}(t)} \mathbb{H}(\mathbf{q}(t), \mathbf{u}(t), \boldsymbol{\Lambda}(t)).
\ee
The first Hamilton equation of motion re-states the dynamics $\dot{\mathbf{q}} = \mathbf{f}$. The second Hamilton equation, obtained after integration by parts, governs the Lagrangian multiplier $\boldsymbol{\Lambda}(t)$. Thanks to this formulation, in the last equation, we can optimize $\mathbf{u}(t)$ pointwise in time, even though its impact on $\mathbf{q}$ and $\boldsymbol{\Lambda}$ must be  accounted through forward and backward integration. We consider problems with free end points. Therefore the costate vector $\boldsymbol{\Lambda}$ also needs to satisfy a transversality condition at $t=T_f$ given by \cite{BookArnoldClassical, pontryagin2018mathematical}
\begin{equation}
    \boldsymbol{\Lambda}(T_f)=-\frac{\partial J_{T_f}}{\partial \mathbf{q}_f},
\end{equation}
where $\mathbf{q}(T_f)=\mathbf{q}_f$. 

We now discuss two ways to utilize this structure in the context of controlling an open quantum system, subject to controlled measurement or dissipation.

\subsection{Lindblad dynamics}\label{sec:4.1}
We first formulate the optimal control problem and the optimality condition for quantum systems under the unconditional continuous measurement process described by  Lindblad dynamics \eqref{eq:lindblad} \cite{Koch_2016, PhysRevApplied.16.054023, Gautier_2025}.

\subsubsection{\label{OC_costate}Optimal final state problem}
In the optimal final state problem for Lindblad dynamics, we ask  what is the optimal schedule  $\theta^{\star}(t)$ that minimizes the terminal cost function at the fixed final time $T_f$:
\begin{equation}\label{eq:terminal_cost}
    J = J_{T_f} = -\mathrm{Tr}(\hat{\rho}(T_f)\hat{\Pi}_0(\pi/2)),
\end{equation}
where $\hat{\rho}(T_f)$ is the final state evolved from a fixed initial state $\hat{\rho}_0$ according to the Zeno dragging Lindblad dynamics \eqref{eq:lindblad} for a fixed total dragging time $T_f$, and $\hat{\Pi}_0(\pi/2)$ is the projector onto solution states in the computational basis.

The PMP \cite{PMP_tutorial, pontryagin2018mathematical} gives the necessary optimality conditions that the system has to satisfy: We define a Pontryagin control Hamiltonian $\mathbb{H}_c$
\begin{equation}\label{eq:PMP_H_Lind}
\mathbb{H}_c(t) = \mathrm{Tr}\{\hat{\Lambda}(t) \;\mathcal{L}(\theta(t))[\hat{\rho}(t)]\}.
\end{equation}
We note that $\mathbb{H}_c$ is just a scalar function of time, and should not be confused with the quantum Hamiltonian driving coherent evolutions. Here $\hat{\Lambda}(t)$ is the costate operator (unlike the co-state vector in \eqref{H_vec}) and plays the role of Lagrangian multipliers. See more discussion about the role of the costate in Appendix \ref{appendix:costate}. The PMP states that the optimal schedule $\theta^\star(t)$ along with the trajectories of state $\hat{\rho}(t)$ and costate $\hat{\Lambda}(t)$ under this schedule should satisfy the following set of coupled equations
\begin{equation}\label{eq:eom_rho_lind_opt_final}
    \frac{d\hat{\rho}}{dt} = \frac{\partial \mathbb{H}_c}{\partial \hat{\Lambda}} = \mathcal{L}(\theta^{\star}(t))[\hat{\rho}], \quad \hat{\rho}(0) = \hat{\rho}_0.
\end{equation}
\begin{equation}\label{eq:eom_costate_lind_opt_final}
    \frac{d\hat{\Lambda}}{dt} = - \frac{\partial \mathbb{H}_c}{\partial \hat{\rho}} = -  \mathcal{L}(\theta^{\star}(t))[\hat{\Lambda}], \quad \hat{\Lambda}(T_f) = -\frac{\partial J_{T_f}}{\partial \hat{\rho}_{T_f}} = \hat{\Pi}_0(\pi/2)
\end{equation}
\begin{equation}\label{eq:pmp_theta_lind_opt_final}
    \theta^\star(t) =  \mathop{\arg\max}\limits_{\theta(t)} \mathbb{H}_c(\hat{\rho}(t), \hat{\Lambda}(t), \theta(t)),  \quad \forall ~t \in[0, T_f].
\end{equation} 
The final time condition on the costate operator in \eqref{eq:eom_costate_lind_opt_final} is the transversality condition. The condition arises since the final state $\hat{\rho}(T_f)$ is free. We call this problem defined above the  \textit{Lindblad optimal final state} (Lindblad-OFS) problem.

It is possible to solve \eqref{eq:eom_rho_lind_opt_final} - \eqref{eq:pmp_theta_lind_opt_final} numerically in some simplified situations, for example, when $
\theta$ is unconstrained. In this case, we can replace \eqref{eq:pmp_theta_lind_opt_final} with the following
\begin{equation}\label{eq:dH_dtheta_lindblad}
    \frac{\partial \mathbb{H}_c}{\partial \theta^{\star}} = \mathrm{Tr}\left(\hat{\Lambda} \frac{\partial \mathcal{L}(\theta^{\star}(t))[\hat{\rho}]}{\partial \theta^{\star}} \right) = 0, \quad \forall t \in[0, T_f].
\end{equation}
This allows gradient-based algorithms, such as Nesterov-GRAPE which we use in this work, to implement the numerical optimization \cite{GRAPE, Nesterov, nesterov2013introductory}. The essential idea of gradient-based algorithms is to iteratively update $\theta(t)$ by
\begin{equation}\label{eq:gradient_descent_lindblad}
    \theta_{k+1}(t) = \theta_{k}(t) + \eta\, \mathrm{Tr}\left(\hat{\Lambda} \frac{\partial \mathcal{L}(\theta_k(t))[\hat{\rho}]}{\partial \theta_k} \right),
\end{equation}
where $\eta$ is some learning rate.

\subsubsection{Optimal time problem}\label{sec:opt_time}
The most common time optimal problem deals with the following task: given the initial state $\hat{\rho}(0) = \hat{\rho}_0$ and (the projector onto) the target final state $\hat{\Pi}_0(\pi/2)$, find the minimal dragging time $T_f$ and the corresponding optimal schedule $\theta^{\star}(t)$ for $t \in [0, T_f]$ to reach  $\hat{\Pi}_0(\pi/2)$ from $\rho(0)$ \cite{PhysRevLett.96.060503}. However, in our case under Zeno dragging dynamics, the solution state is not reachable exactly in finite time. Hence we instead look for an alternative formulation of the optimal time problem.

In settings with algorithmic purposes such as quantum annealing and adiabatic quantum computation, the time-to-solution (TTS) is often the performance metric that is desired to be optimized \cite{speedup_science, speedup_prx}. In the simplest form, the TTS can be defined to be the expected total runtime in order to find the solution, which is $\frac{T_f}{\mathrm{Tr}(\hat{\Pi}_0(\pi/2) \hat{\rho}(T_f))}$ with $T_f$ being the single-shot runtime.  Inspired by this, here we propose to study an optimal time problem with the following cost function to be minimized
\begin{equation}\label{eq_optimal_time_cost}
    J = J_{T_f} = \mathrm{log}(T_f + \tau_m) -  \mathrm{log}(\mathrm{Tr}(\hat{\Pi}_0(\pi/2) \hat{\rho}(T_f))),
\end{equation}
where $\tau_m$ is some regularization constant, which avoids the trivial global minimum corresponding to $T_f \rightarrow 0$. Alternatively, it can be interpreted as the initial state preparation time and the final readout measurement time of a quantum algorithm in a single shot run \cite{zhang2024_ksat}.

The time optimal control problem now requires optimization of both $T_f$ and $\theta(t),\, t\in[0 ,T_f]$. The necessary optimality condition for the optimal control $\theta^{\star}(t)$, the state $\hat{\rho}(t)$ and costate $\hat{\Lambda}(t)$ satisfy the same set of equations as \eqref{eq:eom_rho_lind_opt_final} - \eqref{eq:pmp_theta_lind_opt_final} except that the terminal condition for $\hat{\Lambda}(T)$ is now instead $\hat{\Lambda}(T_f) = - \frac{\partial J_{T_f}}{\partial \hat{\rho}(T_f)} = \frac{\hat{\Pi}_0(\pi/2)}{\mathrm{Tr}(\hat{\Pi}_0(\pi/2) \hat{\rho}(T_f))}$. Besides, we have the additional optimality condition $\frac{dJ_{T_f}}{dT_f} = \frac{\partial J_{T_f}}{\partial T_f} - H_c(T_f) = 0$ for $T_f$, which results in
\begin{equation}
    \frac{1}{T_f + \tau_m} - \frac{\mathrm{Tr}\left\{\hat{\Pi}_0(\pi/2) \mathcal{L}(\theta^{\star}(T_f))[\hat{\rho}(T_f)]\right\}}{\mathrm{Tr}\{\hat{\Pi}_0(\pi/2) \hat{\rho}(T_f)\}} = 0.
\end{equation}
We refer to the problem defined above as \textit{Lindblad optimal time} (Lindblad-OT) problem.

\subsection{The Most Likely Path}\label{sec:4.2}

Having formulated the optimal control problems for the unconditional Lindblad dynamics, we now introduce the concepts and build a contrasting framework for optimal control problems with the conditional dynamics given by \eqref{eq:strat_SME}.
This has the advantage of working in a smaller coordinate space, by virtue of using pure states instead of mixed states, but this comes at the expense of having to manage stochastic readout parameters in the dynamics. 

For example, for the purpose of Zeno dragging, one can approach arbitrarily close to the target state by following a trajectory with arbitrarily large $r_{\alpha}(t)$, regardless of the schedule of $\theta(t)$. 
However, this infinitely fast dynamics happens with a vanishingly small probability density. On the other hand, it makes no sense to trade mixed states for pure states at the expense of following the probability of each possible readout sequence $\{r_{\alpha}(t) \}_{\alpha=1}^m$ in \eqref{eq:strat_SME}. To resolve these issues, we consider \textit{the most likely path optimal control problem} by taking into account the probability densities of the quantum trajectories in the framework of CDJ stochastic path integral, noting that, CDJ optimizations have been shown to be able to boost the probability of trajectories reaching target states ~\cite{CDJ_origin, CJ, karmakar2025cdj, lewalle2024prxq, kokaew2024}. While this will not improve the average performance under purely passive measurements -- indeed, the latter is just equal to the performance of the unconditional case, it paves the way for improving the average performance by measurement-based feedback actions.

\subsubsection{Optimal final state problem}\label{sec:4.2.1}

We consider the probability density functional $\mathcal{P}$ of following some particular path $\{\hat{\rho}(t),  r_{1}(t), r_2(t), \allowbreak \ldots, r_m(t)\}$  during the evolution from $t = 0$ to $t = T$ for the system \eqref{eq:strat_SME}. 

It is helpful to first consider this problem in discrete time, closely following the original formulation by CDJ \cite{CDJ_origin}.  
Consider a sequence of states $\lbrace \hat{\rho} \rbrace$ and readouts $\lbrace \mathbf{r} \rbrace$ at discrete time slices indexed by $j$. The joint probability density for these states and readouts goes like
\be 
\mathcal{P}(\lbrace \hat{\rho} \rbrace,\lbrace \mathbf{r} \rbrace) = \mathcal{P}(\hat{\rho}_i,\hat{\rho}_0)\mathcal{P}(\hat{\rho}_f,\hat{\rho}_N) \prod_{j=0}^{N-1}\mathcal{P}(\hat{\rho}_{j+1},\mathbf{r}_j|\hat{\rho}_j) = 
\mathcal{P}_i\,\mathcal{P}_f \prod_{j=0}^{N-1}\mathcal{P}(\hat{\rho}_{j+1}|\hat{\rho}_j,\mathbf{r}_j)\mathcal{P}(\mathbf{r}_j|\hat{\rho}_j).
\ee
Here $\mathcal{P}_i^{(0)}$ and $\mathcal{P}_f^{(N)}$ provide boundary conditions, such that e.g.~$\mathcal{P}_i^{(0)} = \delta(\hat{\rho}_0 - \hat{\rho}_i)$ denotes perfect preparation of the state $\hat{\rho}_0$ at $\hat{\rho}_i$, and $\mathcal{P}_f^{(N)}$ may similarly denote post--selection on a particular state or distribution. 
Given a state update of the form $\hat{\rho}_{j+1} = \mathcal{E}(\rho_j,\mathbf{r}_j) = \mathcal{E}_j$, such as \eqref{Born_rule}, we have a deterministic state update $\mathcal{P}(\hat{\rho}_{j+1}|\hat{\rho}_j,\mathbf{r}_j) = \delta(\hat{\rho}_{j+1} - \mathcal{E}_j)$ given the readout, and a distribution $\mathcal{P}(\mathbf{r}_j|\hat{\rho}_j)$ describing the stochastic nature of the measurement record itself.
Using the Fourier form of the $\delta$--functions, we may exponentiate the product above (introducing the co-states $\hat{\Lambda}_j$), i.e.
\be 
\mathcal{P}(\lbrace \hat{\rho} \rbrace,\lbrace \mathbf{r} \rbrace) \sim  \exp\left[\ln\,\mathcal{P}_i^{(0)} + \ln\,\mathcal{P}_f^{(N)} - \sum_{j=0}^{N-1} \mathrm{Tr}\left(\hat{\Lambda}_j\cdot (\hat{\rho}_{j+1} - \mathcal{E}_j)\right) - \ln\,\mathcal{P}(\mathbf{r}_j|\hat{\rho}_j)\right]. 
\label{CDJ-discrete}\ee

In the time--continuum limit, and assuming a Gaussian measurement apparatus, the dynamics become like \eqref{eq:strat_SME}, and the readout distribution may be approximated as well  \cite{Philippe_Thesis, Cylke_2025inprep}.
The path integral becomes
\begin{equation}
\begin{aligned}
    \mathcal{P}(\hat{\rho}, r_{1}, \cdots,  r_m) &= \int \mathcal{D}[\hat{\Lambda}] e^{-\mathcal{S}} = \int \mathcal{D}[\hat{\Lambda}] e^{\mathcal{B}-\int_0^T dt\{\mathrm{Tr}(\hat{\Lambda} \frac{d\hat{\rho}}{dt}) - \mathbb{H}_{CDJ}\}},
\end{aligned}
\end{equation}
where $\mathcal{B}$ shorthands the boundary terms. The interpretation of the path integral $\int{\mathcal{D}[\hat{\Lambda}]}$ can be made clear by parameterizing the costate $\hat{\Lambda}$ in coordinates of some basis operators. This is discussed in more detail in Appendix \ref{appendix:costate}. The stochastic action is $\mathcal{S} = \int_0^T dt \{\mathrm{Tr}(\hat{\Lambda} \frac{d\hat{\rho}}{dt}) - \mathbb{H}_{CDJ}\}$, with the CDJ-stochastic Hamiltonian being
\begin{equation}\label{eq:CDJ_H}
\begin{aligned}
    &\mathbb{H}_{CDJ}(\hat{\rho}, \hat{\Lambda}, r, \theta)= \frac{1}{m}\sum_{\alpha = 1}^m \left\{\mathrm{Tr}(\hat{\Lambda}\mathcal{F}_{\alpha}(\hat{\rho}, r_{\alpha}, \theta))- \frac{1}{2}\left[ r_\alpha - 2 \mathrm{Tr}(\hat{\rho} \hat{P}_{\alpha}(\theta)) \right]^2 - g(\hat{\rho}, \hat{P}_{\alpha}(\theta))\right\},
\end{aligned}
\end{equation}
where 
\begin{equation}
\begin{aligned}
    g(\hat{\rho}, \hat{P}_{\alpha}(\theta)) &= 2 \left( \mathrm{Tr}\left\{\hat{\rho} \hat{P}^2_{\alpha}(\theta)\right\} - \mathrm{Tr}\left\{\hat{\rho}\hat{P}_{\alpha}(\theta)\right\}^2\right) = 2 \mathrm{Var}[\hat{P}_{\alpha}(\theta)].
\end{aligned}
\end{equation}
The action $\mathcal{S}$ defined above specifies the log-probability density of any quantum trajectory. The term  $g+\frac{1}{2}[ r_\alpha - 2 \mathrm{Tr}(\hat{\rho} \hat{P}_{\alpha}(\theta)) ]^2$ represents the log-probability density of any conditional dynamics trajectory satisfying \eqref{eq:strat_SME}, and can be interpreted as a naturally arising running cost \cite{lewalle2024prxq, Cylke_2025inprep, Karmakar_2022, barchielli2009quantum}.

In short, we recognize that we again have a problem in the form \eqref{PMP_traditional} \cite{lewalle2024prxq}, where the physical dynamics are now given by \eqref{eq:strat_SME}, and the measurement statistics provide us with particular forms of the running cost $J_t$. Moreover, the final boundary condition $\ln\,\mathcal{P}_f$ can be associated with a terminal cost (such that post--selection on a narrow Gaussian may be understood to be equivalent to imposing a steep quadratic cost on landing near a particular final state).

For any given schedule $\theta(t)$,  the most likely path then corresponds to the trajectory that minimizes the stochastic action $\mathcal{S}$ with respect to the measurement results $r_{\alpha}(t)$. Since $r_{\alpha}$ is unbounded, the minimization is achieved by $\delta \mathcal{S} = 0$, which leads to a set of Hamilton equations \cite{CDJ_origin}
\begin{equation}
    \frac{d\hat{\rho}}{dt}=\frac{\partial\mathbb{H}_{CDJ}}{\partial \hat{\Lambda}}, \quad \frac{d\hat{\Lambda}}{dt}=-\frac{\partial\mathbb{H}_{CDJ}}{\partial \hat{\rho}},\quad\frac{\partial \mathbb{H}_{CDJ}}{\partial r^{\star}_{\alpha}}=0, \forall \alpha \in [m],
\end{equation}
with initial and terminal condition $\hat{\rho}(0) = \hat{\rho}_0$ and $\hat{\Lambda}({T_f}) = 0$. The trajectory specified by the above set of coupled equations is then the most-likely path, which specifies the highest probability trajectory, starting from $\hat{\rho}_0$, generated from the stochastic dynamics \eqref{eq:strat_SME} with control schedule $\theta(t)$. 

Within the context of Zeno dragging, we may consider post-selecting the final state on a region close to the target subspace $\hat{\Pi}_0(\pi/2)$ such that the probability of reaching this region is maximized by CDJ optimization. As discussed before, this can be done by specifying a boundary term
\begin{equation}
    \mathcal{B}  = \mathrm{Tr}(\hat{\rho}(T_f)\hat{\Pi}_0(\pi/2)).
\end{equation}
This is equivalent to including a terminal cost function $J_{T_f} = - \mathcal{B}$ to the total cost function $J$, which now represents the log-probability density of obtaining the post-selected states around the solution subspace at the final time $T_f$, with the dynamics driven by the schedule $\theta(t)$ and measurement records $\{ r_{\alpha}(t)\}_{\alpha=1}^m$. In other words, within the usual framework of optimal control theory, we may now wish to optimize control parameters $\theta(t)$ and $\{ r_{\alpha}(t)\}_{\alpha=1}^m$ in order to minimize a total cost function
\begin{equation}\label{eq:mlp_intro_cost}
\begin{aligned}
   J&=  \int_0 ^{T_f} dt \left\{\frac{1}{m}\sum_{\alpha=1}^m\left(\frac{1}{2}\left[ r_\alpha - 2 \mathrm{Tr}(\hat{\rho} \hat{P}_{\alpha}(\theta)) \right]^2  + g(\hat{\rho}, \hat{P}_{\alpha}(\theta))\right) \right\} + J_{T_f}(\hat{\rho}({T_f}), \theta({T_f})),
\end{aligned}
\end{equation}
where the dynamics of the state variable $\hat{\rho}$ is governed by \eqref{eq:strat_SME} with initial condition $\hat{\rho}(0) = \hat{\rho}_0$.  Viewing $\{ r_{\alpha}(t)\}_{\alpha = 1}^m$ as additional control parameters and applying PMP to the above system, we define a control Hamiltonian $\mathbb{H}_c = \mathbb{H}_{CDJ}$ same as \eqref{eq:CDJ_H}.

The optimality condition is then given by the set of equations the state $\hat{\rho}$ and the costate $\hat{\Lambda}$ have to satisfy:
\begin{equation}\label{eq:eom_rho_mlp_opt_final}
    \frac{d\hat{\rho}}{dt} = \frac{\partial \mathbb{H}_c}{\partial \hat{\Lambda}} = \mathcal{F}(\hat{\rho}, \{ r^{\star}_{\alpha}(t)\}_{\alpha = 1}^m, \theta^{\star}(t)), \quad \hat{\rho}(0) = \hat{\rho}_0,
\end{equation}

\begin{equation}\label{eq:eom_costate_mlp_opt_final}
\begin{aligned}
    \frac{d\hat{\Lambda}}{dt} = -\frac{1}{m}\sum_{\alpha = 1}^m (r^{\star}_{\alpha} -1)\left[\hat{P}_{\alpha}(\theta^{\star}(t))\hat{\Lambda} \right.& \left.+ \hat{\Lambda}  \hat{P}_{\alpha}(\theta^{\star}(t))- 2 \mathrm{Tr}\{\hat{P}_{\alpha}(\theta^{\star}(t))\hat{\rho}\}\hat{\Lambda} - 2 \left(\mathrm{Tr}\{\hat{\Lambda} \hat{\rho}\}-1\right)\hat{P}_{\alpha}(\theta^{\star}(t)) \right] ,\\
    & \hat{\Lambda}(T_f) = -\frac{\partial J_{T_f}}{\partial \hat{\rho}(T_f)} = \hat{\Pi}_0(\pi/2),
\end{aligned}
\end{equation}
where $\mathcal{F}$ is defined in \eqref{eq:strat_SME}. Further simplification can be made to formulate \eqref{eq:eom_costate_mlp_opt_final} into the negative conjugate of \eqref{eq:eom_rho_mlp_opt_final}, see Appendix \ref{appendix:Equation_reduction}. Furthermore, the optimality condition also specifies that the optimal controls $\theta^{\star}(t)$ and $\{ r^{\star}_{\alpha}(t)\}_{\alpha = 1}^m$ have to be chosen to be
\begin{equation}
    \theta^{\star}(t), \{ r^{\star}_{\alpha}(t)\}_{\alpha = 1}^m = \mathop{\arg\max}\limits_{\theta(t), \{ r_{\alpha}(t)\}_{\alpha = 1}^m} \mathbb{H}_c(\hat{\rho}, \hat{\Lambda}, \{ r_{\alpha}(t)\}_{\alpha = 1}^m, \theta(t)).
\end{equation}
As $r_\alpha(t) \in \mathbb{R}$ is in general unbounded, the above condition for optimal $\{r^{\star}_{\alpha}(t)\}_{\alpha = 1}^m$ simplifies to
\begin{equation}\label{eq:mlp_intro_r}
    \quad \frac{\partial \mathbb{H}_c}{\partial r^{\star}_{\alpha}}=0, \forall \alpha \in [m],
\end{equation}
which can be analytically solved to give \cite{karmakar2025cdj}:
\begin{equation}\label{eq:opt_readout_analytical}
\begin{aligned}
    &r^{\star}_{\alpha}(t) = 2\mathrm{Tr}(\hat{\rho} \hat{P}_{\alpha}(\theta^{\star}(t)))+  \mathrm{Tr}\left\{\hat{\Lambda}\left(\hat{\rho} \hat{P}_{\alpha}(\theta^{\star}(t)) + \hat{P}_{\alpha}(\theta^{\star}(t))\hat{\rho} - 2 \hat{\rho} \mathrm{Tr}[\hat{\rho} \hat{P}_{\alpha}(\theta^{\star}(t))]\right) \right\}. 
\end{aligned}
\end{equation}
We henceforth refer to the problem defined above through \eqref{eq:mlp_intro_cost} - \eqref{eq:mlp_intro_r} as \textit{the most likely path optimal final state} (MLP-OFS)  problem.

Just like the case in optimal final state problem with Lindblad dynamics, the above system may also be solved using gradient-based numerical algorithms like Nesterov-GRAPE that we use in this study. Similar to \eqref{eq:dH_dtheta_lindblad} and \eqref{eq:gradient_descent_lindblad}, the essential iteration step for the gradient-based algorithm takes the following form:
\begin{equation}\label{eq:gradient_descent_mlp}
\begin{aligned}
    \theta_{k+1}(t) &= \theta_{k}(t) + \eta \left[ \mathrm{Tr}\left(\hat{\Lambda} \frac{\partial \mathcal{F}(\hat{\rho}, \{r_{\alpha}(t)\}_{\alpha=1}^m, \theta_k(t))}{\partial \theta_k} \right)+ \frac{2}{m}\sum_{\alpha=1}^m(r_{\alpha}(t)-1)\mathrm{Tr}\left(\hat{\rho}\frac{\partial \hat{P}_{\alpha}(\theta_k)}{\partial \theta_k}\right) \right],
\end{aligned}
\end{equation}
where $\eta$ is again some learning rate.

\subsubsection{Optimal time problem}\label{sec:4.2.2}
Similar to the time optimal problem defined for Lindblad dynamics, we can define a time optimal problem for the most likely path by considering a cost function of the form \eqref{eq:mlp_intro_cost} with the terminal cost $J_{T_f}$ given in \eqref{eq_optimal_time_cost}. The state $\hat{\rho}$, costate $\hat{\Lambda}$, optimal readouts $\{ r^{\star}_{\alpha} \}_{\alpha =1 }^m$, and the optimal control $\theta^{\star}(t)$ still need to satisfy the necessary optimality conditions \eqref{eq:eom_rho_mlp_opt_final} - \eqref{eq:mlp_intro_r}, except that the terminal condition for $\hat{\Lambda}$ is now instead $\hat{\Lambda}(T_f) = - \frac{\partial J_{T_f}}{\partial \hat{\rho}(T_f)} = \frac{\hat{\Pi}_0(\pi/2)}{\mathrm{Tr}(\hat{\Pi}_0(\pi/2) \hat{\rho}(T_f))}$. The optimality condition $\frac{dJ}{dT_{f}} = \frac{\partial J_{T_f}}{\partial T_f} - H_c(T_f) = 0$ for $T_f$ now reads
\begin{equation}
\begin{aligned}
    0=& \frac{1}{T_f + \tau_m} + \frac{1}{m}\sum_{\alpha=1}^m \frac{1}{2}\left[ r^{\star}_\alpha(T_f) - 2 \mathrm{Tr}\left(\hat{\rho}(T_f) \hat{P}_{\alpha}(\theta^{\star}(T_f))\right) \right]^2 + g(\hat{\rho}(T_f), \hat{P}_{\alpha}(\theta^{\star}(T_f))) \\ 
    &-\frac{\mathrm{Tr}\left\{\hat{\Pi}_0(\pi/2) \; \mathcal{F}\left(\hat{\rho}(T_f),\{r^{\star}_{\alpha}(T_f) \}_{\alpha=1}^m, \theta^{\star}(T_f)\right)\right\}}{\mathrm{Tr}\{\hat{\Pi}_0(\pi/2) \hat{\rho}(T_f)\}} .
\end{aligned}
\end{equation}
We refer to the problem defined here then as \textit{the most likely path optimal time} (MLP-OT) problem. Similar to the interpretation of the most likely path final state problem, the optimal schedule we obtain from this task is a balance between minimizing TTS and the probability with which such trajectory can occur. 

We also point out that in both optimal final state problem and time optimal problem, one can in principle tune the relative weights between the running cost stemming from the CDJ stochastic path integral and the original cost one wishes to minimize in \eqref{eq:mlp_intro_cost}. In this way, one can adjust the importance of the probability density of the quantum trajectory and thus obtain a different optimal control $\theta^{\star}(t)$.

Finally, we note that both Lindblad-OFS and MLP-OFS problems are analytically solvable for single-qubit systems, and their optimal schedules are  linear functions of time.  A pedagogical derivation of these optimal control solutions can be found in Appendix \ref{appendix:analytical_single_qubit} and \cite{lewalle2024prxq}.

\section{Numerical Results of Optimal Schedule for Lindblad Dynamics and the Most Likely Path}\label{sec_numerics}

In this section, we present the numerical optimization results following the framework of optimal control theory we introduced in Section \ref{sec_OCT}. We show results for Lindblad and most likely path dynamics, with the optimization objective being the optimal final state problem and also the time-optimal problem. The results for the most likely path presented here can be viewed as a numerical extension of the methods developed in \cite{lewalle2024prxq}. To capture the essential features of the problem in a simpler setting, we primarily focus on problem instances of $2$-SAT on the ring with a single solution (defined in Appendix \ref{appendix:ksat}), which requires only $m= n+1$ clauses for $n$ qubits. We also include some numerical results for 3-SAT problems with a single solution.

\subsection{Optimal Final State Problem}\label{sec_optimal_final_state_numerics}

For optimal final state problems, we are effectively optimizing the final state fidelity with respect to the target state $f = \mathrm{Tr}[\hat{\rho}(T_f) \hat{\Pi}_0]$. We implement the optimization for both Lindblad dynamics and the most likely path using the Nesterov-GRAPE algorithm. 
The results for $n$-qubit $2$-SAT with $n=2$, defined on a ring with a single solution are shown in Fig.~\ref{fig:_opt_schedule}. The optimized schedules for both dynamics are very similar. For both dynamics, we see that the optimal schedule does not begin at $\theta = 0$ and end at $\theta = \pi/2$. 
Instead, they jump over some regions at the beginning and the end. This can be understood as the following way intuitively: From proposition \ref{proposition_mixing}, we notice that the error associated with inadequate mixing is small at small $\theta$ because the initial fidelity is close to 1. 
Moreover, as we see from the inset on Fig.~\ref{fig:_opt_schedule} (a), the running problem cost operator \eqref{eq:cost_op} has a vanishing gap at the beginning and a large gap at the end, which makes the mixing at small $\theta$ very inefficient. This means spending time here can only reduce the already small error in a very inefficient way. Therefore, the optimal schedule would tend to skip the regions near the $\theta = 0$ where the mixing is very inefficient and the fidelity is high to trade for time spent on efficient mixing later with larger $\theta$. After the spectral gap increased to a reasonably large value, the speed of the schedule starts to follow the inverse of the gap as in Theorem \ref{theorem_convergence}. We notice this is in the similar spirit of the schedule used in \cite{benjamin2017measurement}, where the schedule starts at some non-zero value $\theta(0) > 0$, and then speeds up toward the end. The instantaneous fidelity outperforming the linear schedule confirms the correctness of our calculations.

\begin{figure}[htbp]
    \centering
    \hspace*{-1.5cm}
    \begin{tikzpicture}
        \node at (0,0){\includegraphics[width=1.15\linewidth]{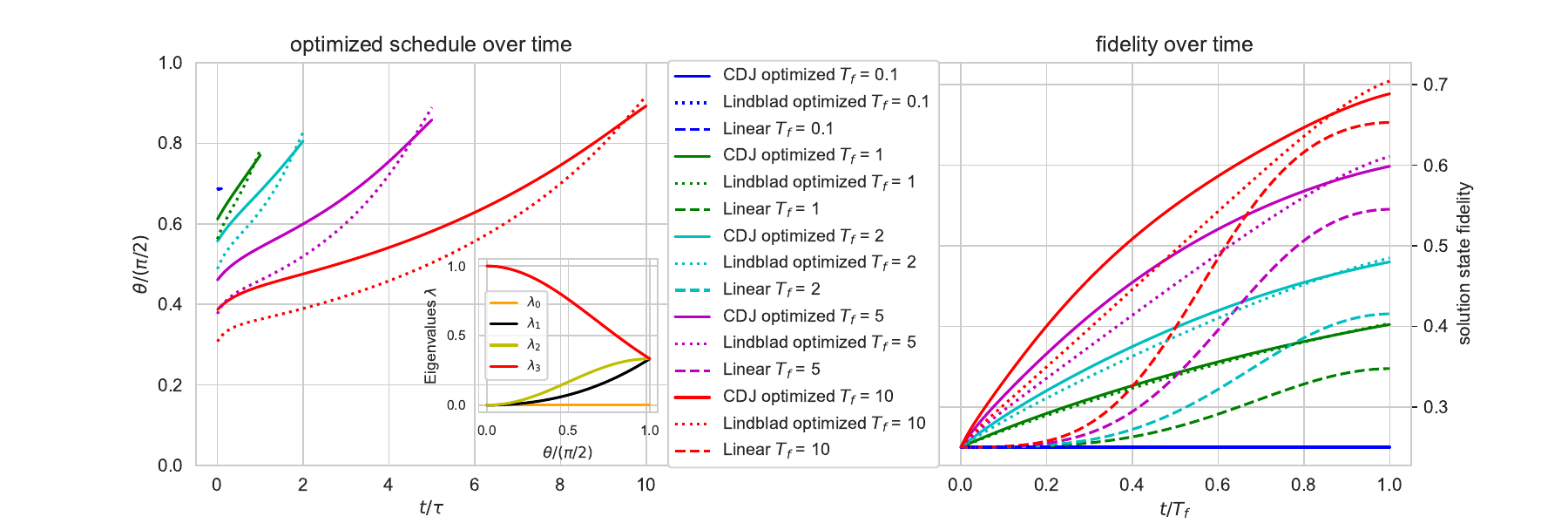}};
        \node at (-6.4, 2) {(a)};
        \node at (2.1, 2) {(b)};
    \end{tikzpicture}
    \caption{(a) Optimal schedule $\theta^{\star}(t)$ from Lindblad-OFS and MLP-OFS. The inset shows the eigenvalues of the cost operator $\hat{O}(\theta) = \frac{1}{m}\sum_\alpha \hat{P}(\theta)$ as a function of $\theta$. (b) instantaneous fidelity to the final solution, from Lindblad dynamics as a function of time, for $n = 2$ 2-SATwith various total dragging times $T_f$ and schedules optimized according to Lindblad-OFS and MLP-OFS, compared to a linear schedule. In (a), we can see that none of the schedules start at $0$ and end at $\pi/2$. Instead, the smaller the total dragging time $T_f$, the shorter the path it sweeps. We can see a monotonically increasing spectral gap between the ground state and the second largest eigenvalue of $\hat{O}(\theta)$. This is consistent with the accelerating behavior we observed in the optimized schedules.  In (b), we can see that the fidelity of the optimized schedule is always better than that under the linear schedule. We also see that fidelity from the Lindblad-OFS optimized schedule outperforms that from the MLP-OFS optimized schedule, as expected when running it on Lindblad dynamics.} 
    \label{fig:_opt_schedule}
\end{figure}

We note that the convergence bound presented in Theorem \ref{theorem_convergence} is generally far from tight, since that derivation follows a non-optimized linear schedule, completely agnostic of the problem instance. Any additional information should allow one to obtain more efficient protocols. Motivated by this, the analysis in this section considers the opposite bound, namely to optimally schedule the Zeno dragging in the hypothetical situation that the solution to the SAT problem is fully known. Note that this is still less extreme than directly rotating each qubit to the solution state, which one could, in principle, also do in this situation. Our rationale for considering this hypothetical optimization situation is to make use of both bounds, the one from the agnostic situation and the one from the full knowledge situation, to extract guidelines and intuition for designing good schedules for realistic settings.

In Fig.~\ref{fig:_opt_schedule} (b), we calculated the fidelity output from the Lindblad dynamics with schedules generated from solutions of the Lindblad-OFS problem as well as the MLP-OFS problem. We see that the fidelities
generated from both versions of the optimal final state problems outperform the fidelity from the linear schedule $\theta(t) = \frac{t}{T_f} \cdot \frac{\pi}{2}$. Meanwhile, the fidelity with  the schedule from Lindblad-OFS is better than that from MLP-OFS as expected, because we are considering the average final state fidelity, which is the final output of Lindblad dynamics.

\begin{figure}[tbhp]
    \centering
        \hspace*{-1.2cm}
        \includegraphics[width=1.1\linewidth]{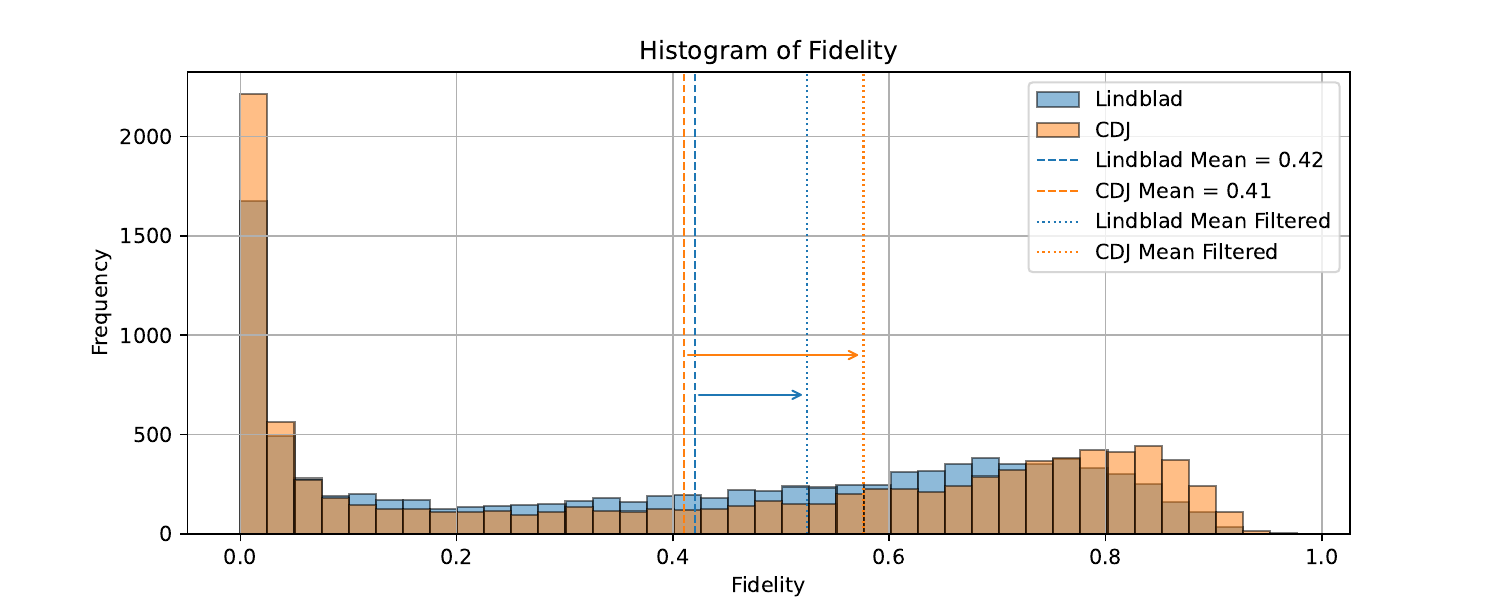}

\caption{Statistics of the final state fidelity for $n=2$ qubit 2-SAT} under conditional clause measurement dynamics \eqref{eq:Kraus}. The results for schedules generated from both Lindblad-OFS and MLP-OFS are shown. We see that the fidelities from the Lindblad-OFS optimized schedule are more evenly distributed, while the fidelities from MLP-OFS are more concentrated around high and low values. The mean of the Lindblad-OFS optimized schedule is higher than that of the MLP-OFS. However, assuming that we could implement some filter equivalent to post-selecting only trajectories above the cutoff $f_{thre} = 0.05$, the resulting new mean of the MLP-OFS optimized schedule is significantly higher than that of the Lindblad-OFS.
\label{fig:histogram}

\end{figure}

However, schedules from the MLP-OFS problem can be better than schedules from  the Lindblad-OFS in slightly modified but practically important situations. Recall that MLP-OFS effectively optimizes the probability of reaching states close to the target \cite{kokaew2024, karmakar2025cdj}. To confirm this intuition, we generated the full distribution of final solution state fidelities, depending on the output of stochastic clause measurement dynamics described by \eqref{eq:Kraus}, instead of just the average dynamics by Lindblad \eqref{eq:lindblad}, for a two-qubit $2$-SAT problem with total dragging time $T_f = \tau$. As shown in Fig.~\ref{fig:histogram}, although the Lindblad-OFS fidelity distribution has  higher mean, the MLP-OFS fidelity distribution exhibits higher variance: more weight near 1 (success) and near 0 (failure). When one has access to the measurement signals, these close-to-zero trajectories are very easy to detect, because they typically have evolved into the undesired subspaces and thus have small overlap with the solution subspace. As studied in \cite{zhang2024_ksat}, one can use a filtering technique to detect the occurrence of such trajectories in the undesired subspaces, and truncate such dynamics early, resulting in effective post-selection.  One of the simplest forms of such filtering can be implemented by time-integrating the readouts using a boxcar or exponential window function, and the trajectory is truncated or discarded when the filtered readouts pass some threshold values. A similar approach has also been used in continuous quantum error correction, where instead of truncating the trajectories, one performs error correcting operations as feedback when a threshold is reached by the filtered signal \cite{CQEC_Juan,PhysRevA.95.032317}. More generally, the MLP-OFS or similar settings would typically be advantageous in the presence of measurement-based feedback, where low fidelity trajectories could be effectively countered. We leave a full study of feedback strategies for future work and here consider the simple feedback action of effectively cutting short the trajectories once the fidelity is consistently lower than $f_{thre} = 0.05$. While we leave open the details of its implementation, and the correspondence should be taken with a grain of salt, the mechanism in practice should amount to cutting short those trajectories where, after a negligible amount of time, one of the clause measurements has conclusively converged to clause failure. As shown in Fig. \ref{fig:histogram}, the mean of MLP-OFS fidelity after post-selection is now higher than that of Lindblad-OFS. This demonstrates the potential superiority of MLP-OFS optimization for heralded dynamics where one can extract extra information from the measurement signals compared to unheralded Lindblad dynamics.

\begin{figure}[tbhp]
    \centering
    \includegraphics[width=\linewidth]{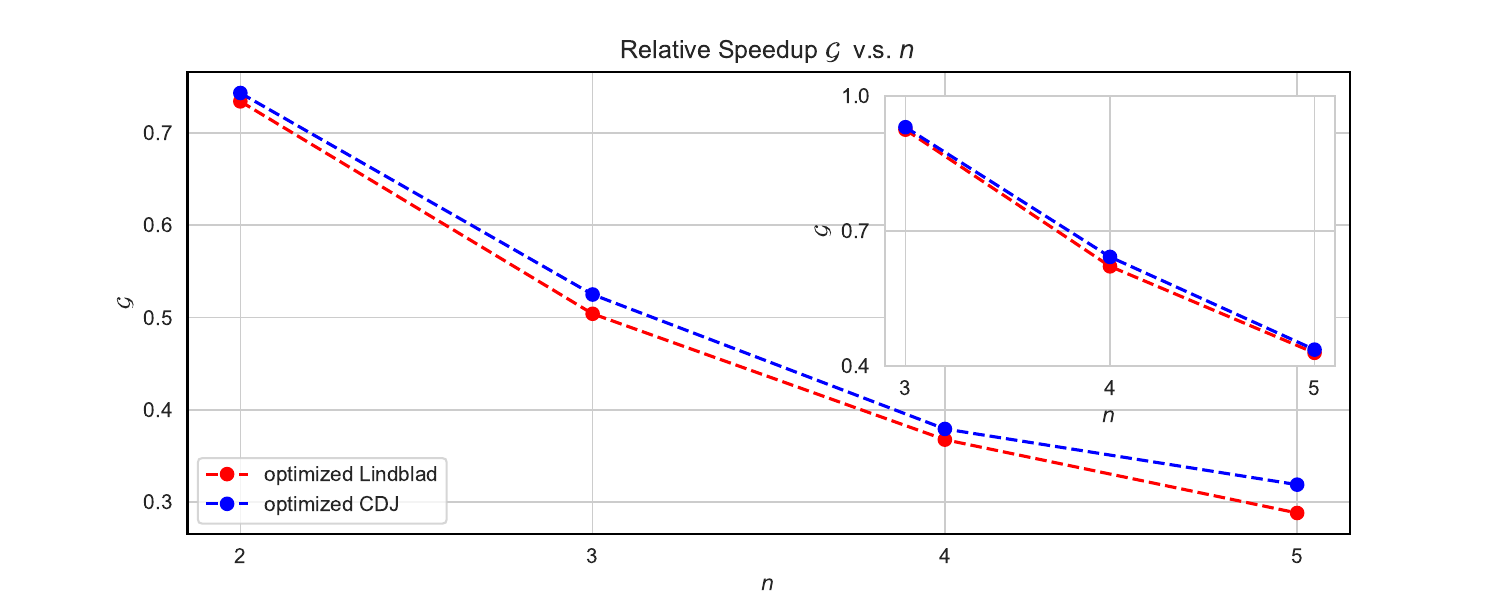}
    \caption{Relative speedup $\mathcal{G}$ (defined in \eqref{eq:relative_gain}) as a function of the number of qubits $n$ for single solution $2$-SAT problems on a ring. The inset shows the same calculation for 3-SAT with a single solution, where for $n=3$ there are $m=7$ clauses, and for $n=4$ and $n=5$ there are $m=\lceil4.26n\rceil$ clauses.  Note, $m=\lceil4.26n\rceil$ marks the transition from solvable 3-SAT instances to mostly unsolvable instances \cite{benjamin2017measurement,achlioptas_rigorous_2005}. The corresponding $\text{TTS}^{opt}$ of optimized Lindblad and optimized CDJ are obtained from the solutions of Lindblad-OT and MLP-OT problems respectively. The optimal TTS of  the linear schedule is obtained from the Lindblad optimal time problem without scheduling optimization, i.e., only single-shot dragging time $T_f$ is optimized. For the other case both $T_f$ and the corresponding schedules $\theta(t)$ are optimized for Lindblad and CDJ respectively. All of the TTS are then calculated from the final states of Lindblad evolutions following their optimal $T_f$ and corresponding schedules $\theta(t)$.}
    \label{fig:time_opt}
\end{figure}

\subsection{Time Optimal Problem}

With these optimized schedules, we are able to obtain the solution to the optimal time problem. The objective we would like to optimize is the (logarithm of) the $\mathrm{TTS} = (T_f + \tau_m)/\mathrm{Tr}[\rho(T_f) \hat{\Pi}_0]$. Here we assume $\tau_m = 5\tau$ for 2-SAT and $\tau_m = 2\tau$ for 3-SAT, which makes sense as the final readout time in the canonical basis needs to be several times larger than $\tau$ in order to obtain accurate information \cite{zhang2024_ksat}. We numerically solve Lindblad-OT and MLP-OT problems by implementing the single-shot dragging time optimizations with schedules from solutions of Lindblad-OFS and MLP-OFS problems defined in Section \ref{sec_OCT}. We also implement $T_f$ optimization for Lindblad dynamics with linear schedule $\theta(t) = \frac{\pi}{2}\cdot \frac{t}{T_f}$, i.e. no schedule optimization. All of the optimal TTS are then obtained from \textit{final states of Lindblad evolutions} with optimal dragging time $T_f$ and corresponding schedules $\theta(t)$. This makes sure the optimal TTS is the expected runtime in order to obtain the solution state. To illustrate the improvement of TTS thanks to schedule optimizations, we define the relative speedup
\begin{equation}\label{eq:relative_gain}
    \mathcal{G} = \frac{ \mathrm{TTS}_{opt}}{\mathrm{TTS}_{linear}},
\end{equation}
where $\mathrm{TTS}_{linear}$ is the optimal TTS obtained from linear schedule, while $\mathrm{TTS}_{opt}$ stands for the optimal TTS obtained from either optimized Lindblad schedules or optimized CDJ schedules. The numerical results for $\mathcal{G}$ with varying qubit number $n$ are shown in Fig.~\ref{fig:time_opt}. We see that both Lindblad and most likely path scheduling optimization can reduce the optimal TTS from the linear schedule significantly, as $\mathcal{G} < 1$. In addition, Lindblad optimization admits better relative speedup than most likely path optimization as expected. However, with the similar intuition developed in Section \ref{sec_optimal_final_state_numerics}, we expect the relative speedup from the most likely path optimization can outperform the relative speedup from Lindblad optimization if the measurement signals are available to produce effective post-selection or other appropriate feedback actions.

\subsection{Clause-Wise Schedule Optimization}
One natural question regarding the schedule optimization is: can we further improve the performance by allowing each qubit $i$ to be associated with its individual control $\theta_i$? Our preliminary investigation indicates that the answer could depend on the problem instance. For example, for Zeno dragging with $2$-SAT problems, we found that $2$-SAT on a ring always admits an optimal schedule such that $\theta_1(t) = \cdots = \theta_n(t)$, i.e.~all qubits should move together. 
This is not surprising due to the symmetry of these problem instances. 
However, for some other problem instances, we found that the optimal schedule in general does not need to be moving together. 
This is illustrated in Fig.~\ref{fig:nq_3_multi_l2_distance}. 
It will be interesting to further explore the relationship between the form of optimal schedules and the structure of the SAT instances, which we leave for future work.

\subsection{Summary of numerical results}

The numerical results presented above show that the qualitative behavior of the solutions for unconditioned (Lindblad) dynamics obtained with the optimized schedules is consistent with the convergence results for unconditioned dynamics presented in Section \ref{sec_convergence}. Specifically, we see that the measurement axis can move faster towards the solution in the region having a higher spectral gap of $\hat{O}(\theta)$ (Fig.~\ref{fig:_opt_schedule}), in agreement with the prediction of Theorem \ref{theorem_convergence}. The numerical examples also show the potential advantage of the most-likely-path-based schedule optimization when given access to heralded individual trajectories (Fig.~\ref{fig:histogram} and Fig.~\ref{fig:time_opt}). This paves the way for future work including feedback strategies based on the Zeno measurement results.

\section{Summary and Discussion}\label{sec_discussion}

\begin{figure}[t]
    \centering

    \begin{tikzpicture}
        \node at (-5,0) {\includegraphics[width=0.52\linewidth]{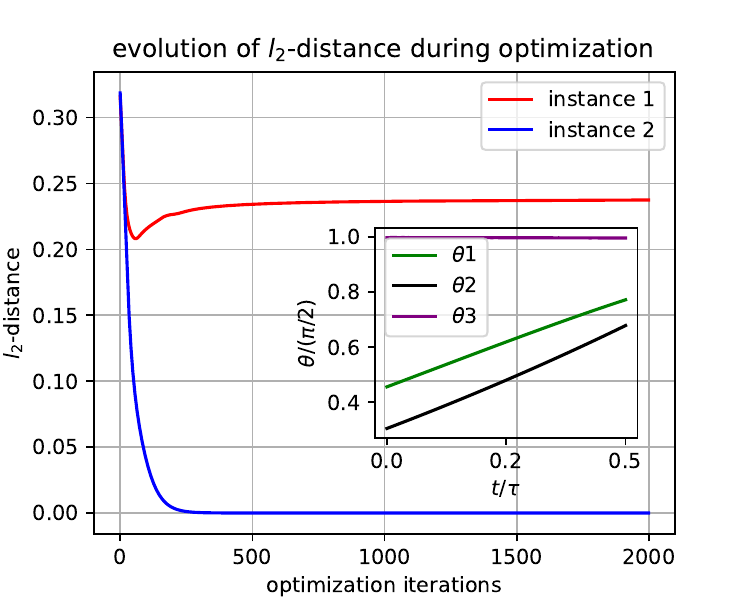}};
        \node at (3,0) {\includegraphics[width=0.52\linewidth]{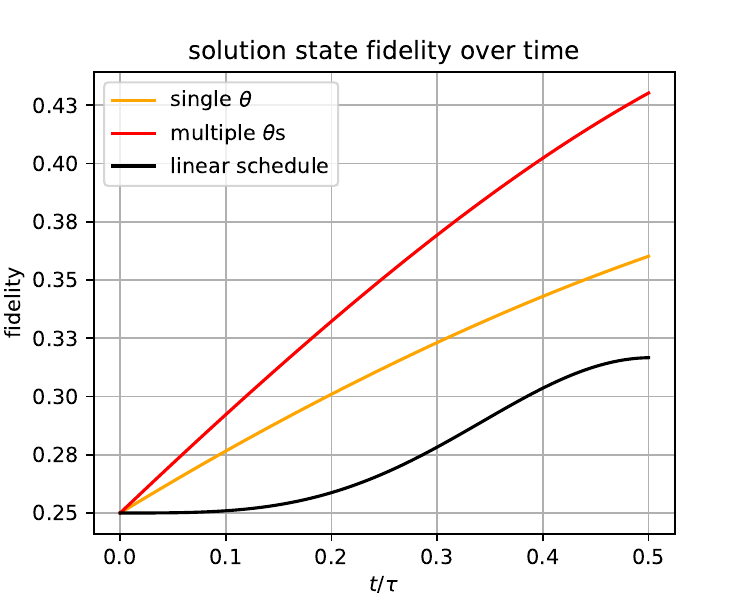}};
        \node at (-8.2,2.75) {(a)};
        \node at (0, 2.75) {(b)};
    \end{tikzpicture}

    \caption{Examples of the optimal schedules when allowing a separate schedule $\theta_i(t)$ for each qubit, for 3-qubit $2$-SAT instances. Here we study two illustrating instances. Instance $1$ is $(\bar{b}_1 \vee b_3 ) \wedge (\bar{b}_2 \vee b_3) \wedge (\bar{b}_1 \vee \bar{b}_3 ) \wedge (b_1 \vee \bar{b}_2 )$, whose solutions are $b_1 = \mathrm{F}$, $b_2 = \mathrm{F}$ and $b_3 = \mathrm{F}$ or $b_1 = \mathrm{T}$, $b_2 = \mathrm{T}$ and $b_3 = \mathrm{F}$; Instance 2 is $(b_1 \vee \bar{b}_2 ) \wedge (b_2 \vee \bar{b}_3) \wedge (\bar{b}_1 \vee b_3 )$, whose solutions are $b_1 = b_2 = b_3 = \mathrm{F}$ or $b_1 = b_2 = b_3 = \mathrm{T}$. Thus for instance 1, the third variable $b_3$ can be any value and is not correlated to $b_1$ and $b_2$, while for instance 2, $b_1$, $b_2$ and $b_3$ are all correlated. In order to characterize the form of the optimal solution, we have defined the $l_2$-distance between the schedules of the different qubits as $\bar{d}_{l_2} = \frac{1}{n(n-1)}\sum_{i\neq j \in [n]}{||\theta_i(t) - \theta_j(t)||_{l_2}}$, where $||f(t)||^2_{l_2} = \int_0^T dt f^2(t)$.  Panel (a) shows the evolution of $\bar{d}_{l_2}$ during the iterations of the Nesterov-GRAPE optimization algorithm. We see that $\bar{d}_{l_2}$ for instance 2 converges to zero indicating all optimal $\theta_i$ are moving together; for instance 1, $\bar{d}_{l_2}$ converges to a non-zero value indicating optimal $\theta_i$ are not moving together. The embedded plot in (a) is the optimized schedule for instance 1, confirming the above claim. Panel (b) shows the solution state fidelity for instance 1 as a function of time for the optimized schedule with 3 $\theta_i$, the optimized schedule with a single $\theta$, and the un-optimized linear schedule with a single $\theta$. We see that the schedule with multiple $\theta_i$ performs the best, which confirms the correctness of our calculations.}
    \label{fig:nq_3_multi_l2_distance}
\end{figure}

In this paper, we have studied the dynamics of multi-channel Zeno dragging under generalized measurements with varying strength; such schemes necessarily rely on measuring non-commuting observables weakly and simultaneously. 
We have proven the convergence condition as a form of adiabatic theorem for measurement-driven dynamics, showing how the spectral gap $G(\theta)$ of the cost operator $\hat{O}(\theta)$
characterizes the speed of modulation of $\theta(t)$ in order to remain in the desired subspace. 
The results also explain the empirical behavior observed in previous work introducing the weak measurement driven approach to solving $k$-SAT problems \cite{zhang2024_ksat}. 
We have further introduced a framework for schedule optimization for both the unconditional dynamics setting and the conditional dynamics setting from the perspective of optimal control theory, based on ideas in \cite{lewalle2024prxq}. 
This optimized dynamics effectively sets the lower bound for the convergence time: indeed, to be practical, any strategy for constructing the dragging schedule must moreover have lower complexity than solving the $k$-SAT problem itself.

We obtain numerical results that both confirm the intuition implicit in our theoretical convergence bounds and demonstrate that scheduling optimization can significantly improve the performance of finite--time Zeno dragging.

A potential strategy to improve Zeno dragging is feedback, where we apply additional unitaries conditioned on the (history of) measurement records. 
Quantum feedback control has been used to stabilize target states, speedup purification, achieve non-Hermitian dynamics, design quantum algorithms, and construct time--continuous error correction protocols \cite{feedback_review, CQE1, Ahn_2003, Ahn_2004, Sarovar_2004, vanhandel2005optimal, Oreshkov_2007, wiseman2009quantum,  Mascarenhas_2010, CQEC_Juan, Mohseninia2020alwaysquantumerror, Atalaya_2021_CQEC, Convy2022logarithmicbayesian, Livingston_2022, Convy_2022, Magann_prl, rapid_purification, karmakar2025noisecancelingquantumfeedbacknonhermitian}. 
As has been demonstrated in both classical and quantum algorithmic applications, feedback can be crucial to achieve computational efficacy \cite{Schoning,quantum_schoning, benjamin2017measurement}. 
A valuable future direction is then applying ideas of feedback to Zeno dragging, such that the probability of staying in the desired subspace can be improved.

We note that at this stage our primary focus is on analyzing and characterizing the open-loop performance of the Zeno dragging approach to combinatorial optimization. Before comparing the performance of the Zeno dragging approach with established state-of-the-art optimization algorithms using digital quantum circuits, it is necessary to add local feedback correction strategies, which constitute a natural extension of continuous measurement quantum protocols \cite{feedback_review} and is also done in classical algorithms like that of Schöning \cite{Schoning}. It will also be useful to explore integration of the specific advantages of the Zeno dragging approach with continuous quantum error correction \cite{SHG_prl, blumenthal2022demonstration, CQE1, Atalaya_2021_CQEC, Livingston_2022, Convy_2022}.

We also comment on the applicability of the optimization framework, which in general needs numerical calculations. The computational expense of optimization tailored to specific problem instances will in general limit its practicality for larger systems. 
Indeed, solving these optimization problems is often even harder than the classical problem encoded in the Zeno dragging protocol to begin with. More practical schemes could follow two routes. A first idea would be to extract insights from low-dimensional optimal control results, in order to deduce guidelines for efficiently building sub-optimal but significantly improved schedule on actually interesting problems. Future work along these lines comparing $k$-SAT problem structure to the resulting optimal controls on small problems may be insightful.
A second idea could be to use our analytical results from Theorem \ref{theorem_convergence} as a reduced dynamical system, if we can efficiently estimate (useful bounds on) the spectral gap as a function of $\theta$ or of the $\theta_i$; following the line similar to \cite{Brachistochrone}, this would allow to optimize schedules with lower accuracy but at a much lower cost.
We remark that a similar approach has been applied in the projective version of the BZF $k$-SAT algorithm \cite{benjamin2017measurement}.
 
Finally, we discuss possible strategies to deal with noise in the circuit realization of Zeno dragging with generalized measurement, as shown in Fig.~\ref{fig:collision_1Q}. 
Our consideration of Zeno dragging is so far restricted to the noiseless case, which is only idealized. 
In reality, many forms of errors would occur in the device, such as readout noise during the measurements and stochastic noise on the gates. 
However, due to the nature of QZE, one could expect natural robustness of Zeno dragging protocols against certain experimental imperfections, with an intuition similar to that underpinning autonomous quantum error correction (AQEC) \cite{PhysRevA.90.062344, mirrahimi2014dynamically, Lihm_2018, Guillaud_2019, lebreuilly2021autonomous, gertler_protecting_2021, xu_autonomous_2023, lewalle2024prxq}. The basic idea of AQEC is that always--on dissipation can be engineered to reduce the rates of certain logical errors occurring, e.g.~by monitoring or dissipating stabilizer observables; this is closely connected to continuous QEC. In addition, since the QZE scheme by definition requires many rounds of mid-circuit measurements, it could support error-correcting feedback actions and measurement-based error mitigation in a natural way \cite{QEM_RMP, ren2025error, PhysRevResearch_meta_shadow, self2024protecting, hu2025demonstration}. We consequently believe that there are many avenues for future work developing the ideas in this paper, and incorporating them into diverse applications across the quantum information science ecosystem.

\begin{acknowledgements}
    This work has been supported by the U.S. Department of Energy, Office of Science, National Quantum Information Science Research Centers, Quantum Systems Accelerator.
    AS acknowledges the financial support of Plan France 2030 through the project ANR-22-PETQ-0006.
    The authors greatly thank Pierre Rouchon and Artem Mamichev for helpful discussions, and are grateful to Torin Stetina for early conversations related to this work.
    PL is grateful to the UMass Lowell department of Physics $\&$ Applied Physics for their hospitality during part of this manuscript’s preparation. KBW acknowledges discussions at the Simons Institute for the Theory of Computing in Berkeley that stimulated the development of this work. 
\end{acknowledgements}

\appendix

\section*{Appendix}
\addcontentsline{toc}{section}{Appendix}

\section{A Brief Introduction to Boolean Satisfiability Problems}\label{appendix:ksat}
Here we briefly review the Boolean satisfiability problems (SAT), and describe the problem instances we used in this work. A SAT problem is specified by $n$ Boolean variables $\{b_i \}_{i=1}^n$, each of which can take values of either `TRUE' or `FALSE' (henceforth T or F), and $m$ clauses $\{C_i\}_{i=1}^m$. Boolean variables appear in clauses in the form of literals $x_i \in \{b_i, \bar{b}_i \}$, where $\bar{b}_j$ is the negation of $b_j$. If every clause contains $k$ Boolean variables, then it is a $k$-SAT problem. A clause is evaluated to be T if at least one of the literals involved is T. Thus, a clause $C_i$ in $k$-SAT can be written as literals connected by logical `OR', also denoted '$\vee$'; for example, with $k=3$ we can have
\begin{equation}
    C_i = b_2 \vee b_3 \vee \bar{b}_4,
\end{equation}
which returns T if $b_2 = \text{T}$ or $b_3 = \text{T}$ or $b_4 = \text{F}$. A $k$-SAT problem asks if there exists an assignment $\vec{b} \in \{\text{T}, \text{F} \}^n$ such that the Boolean formulae in the conjunctive normal form (CNF)
\begin{equation}
    F_{CNF} = C_1 \wedge C_2 \wedge \cdots \wedge C_m
\end{equation}
is evaluated to be T, i.e., if all the clauses can be satisfied. In other words, $F_{CNF}$ involves the logical `AND' of all the clauses. This definition of $k$-SAT makes a clear one-to-one mapping between each clause $C_\alpha$ in  $k$-SAT and the projector $\hat{P}_\alpha(\theta)$ defined in \eqref{eq:projector}, where each literal $l_{\alpha_i}$ is equivalently defined to take values $\{-1, +1 \}$ depending on whether the Boolean variable appears in the negated form (-1) or the positive form (+1).

For $k=2$, a $2$-SAT problem can be efficiently solved in linear time. For $k\geq 3$, $k$-SAT becomes NP-complete and is not expected to be solved efficiently \cite{Cook_1971, Karp_1972, Levin_1973}, neither on classical nor on quantum computers. 
Despite this computational complexity gap, we take $2$-SAT as the demonstrating example in this work, as it keeps the problem simple to describe and carries the essential structure when solved by Zeno dragging algorithms. 
In several examples, we will further focus on the $2$-SAT problems ``defined on a ring'', i.e.~whose Boolean formula takes the following form:
\begin{equation}
    F_{CNF} = (b_1 \vee \bar{b}_2) \wedge (b_2 \vee \bar{b}_3) \wedge \cdots  \wedge (b_{n-1} \vee \bar{b}_n) \wedge (b_n \vee \bar{b}_1).
\end{equation}
This family of problems has exactly 2 solutions, $b_1 = b_2 = \cdots = b_n = T$ or $b_1 = b_2 = \cdots = b_n = F$. The symmetry of this setting under index shift, $i\mapsto i+k \; \mathrm{modulo}(n)$ for all $i$, allows for somewhat further analysis. Alternatively, we can add another single clause, e.g. $C_* = (b_1 \vee b_2)$ to $F_{CNF}$ such that one of the solutions is excluded, resulting in a $2$-SAT instance with a single solution.

\section{Analysis of Lindbladian Zeno dragging}

We here provide the technical proof of the results in Section \ref{sec_convergence}. We will assume the quantum system consists of $n$ qubits, and the cost operator $\hat{O}(\theta)$ admits a spectral decomposition $\hat{O}(\theta) = \sum_{i = 0} ^{2^n-1} \lambda_i |\lambda_i \rangle \langle \lambda_i |$. We further assume that $\hat{O}(\theta)$ has at least one 0-eigenvalue ground state, associated to $\lambda_0=0$. We recall that $\hat{\Pi}_0$ denotes the orthogonal projector onto the 0-eigenspace and thus $\mathrm{Tr}[\hat{\Pi}_0]$ denotes its dimension. We assume that this dimension does not depend on $\theta$. This is the case, for instance, as soon as $\theta_i > 0$ for the $k$-SAT setting \eqref{eq:projector}.

\subsection{Lower bound of the solution state fidelity}\label{appendix:Lower bound of the solution state fidelity}

We first show how the measurement channel retains the fidelity and how the coherence between the ground eigenspace (eigenvalue 0) and the other eigenstates of $\hat{O}(\theta)$ decays under the action of $\mathcal{T}(\theta)$. In particular, we have the following lemma.

\begin{restatable}[]{lemma}{applemma}\label{lemma_decomposition}
Let a frustration-free operator be $\hat{O}= \frac{1}{m}\sum_{\alpha}\hat{P}_{\alpha}$, where each $\hat{P}_{\alpha}$ is some projector. Denote the eigenvalues and eigenstates of $\hat{O}$ as $0 \leq \lambda_i \leq 1$ and $|\lambda_i\rangle$ respectively for $i=0,1,...,\, 2^n-1$. 
Write any density matrix as 
    $$\hat{\rho} = f \, \hat{\rho}_0 + (1-f) \hat{\rho}_{\perp} + \gamma \, (\hat{c}+\hat{c}^\dagger),$$ 
where $\hat{\rho}_0$ has support on the 0-eigenvalue subspace of $\hat{O}$ i.e.~$\hat{\Pi}_0\, \hat{\rho}_0\, \hat{\Pi}_0  = \hat{\rho}_0$; $\hat{\rho}_{\perp}$ has support on the orthogonal subspace i.e.~$\hat{\Pi}_0\, \hat{\rho}_{\perp} = \hat{\rho}_{\perp} \hat{\Pi}_0 = 0$; and $\hat{c}$ denotes coherences between these two spaces, i.e.~$\hat{\Pi}_0 \hat{c} = \hat{c}$, $\hat{c} \hat{\Pi}_0 = 0$, with $\gamma \in \mathbb{R}_{\geq 0}$. Note that for positivity of $\hat{\rho}$, when  $\tr{\hat{\Pi}_0}>1$ the image of $\hat{c}$ must lie in the span of $\hat{\rho}_0$. More precisely, denoting $\hat{\Pi}_0^{(r)}$ the projector onto the span of $\hat{\rho}_0$, thus with $\hat{\Pi}_0^{(r)} \hat{\Pi}_0 = \hat{\Pi}_0 \hat{\Pi}_0^{(r)} = \hat{\Pi}_0^{(r)}$ and $\hat{\rho}_0^{(pi)}$ the pseudo-inverse satisfying $\hat{\rho}_0^{(pi)} \hat{\rho}_0 = \hat{\rho}_0 \hat{\rho}_0^{(pi)} = \hat{\Pi}_0^{(r)}$, we also have $\hat{\Pi}_0^{(r)} \hat{c} = \hat{c}$ and positivity requires $f(1-f) \geq \gamma^2 \, \mathrm{Tr}[\hat{c}^\dagger\hat{\rho}_0^{(pi)} \hat{c}]$.

Then after the action of the quantum channel $\mathcal{T}$, the resulting state is 
    $$\mathcal{T}(\hat{\rho}) = f \, \hat{\rho}_0 + (1-f) \, \hat{\rho}^{\prime}_{\perp} + \gamma^\prime (\hat{c}^\prime + \hat{c}^{\prime \dagger}) \, ,$$
    where $\hat{\rho}^{\prime}_{\perp}$ is some density matrix
    with support only in the subspace orthogonal to $\hat{\Pi}_0$ and $\hat{c}'$ describes coherences between that subspace and the subspace spanned by $\hat{\Pi}_0$. The latter satisfies $\mathrm{Tr}[\hat{c}^{\prime \dagger} \hat{\rho}_0^{(pi)} \hat{c}^{\prime}] =  \mathrm{Tr}[\hat{c}^\dagger \hat{\rho}_0^{(pi)} \hat{c}]$ with nonnegative $\gamma^\prime \leq (1-\beta\, G)\,\gamma$,
    where $G=\min_i \{\lambda_i > 0 \}$ is the gap of operator $\hat{O}$.
\end{restatable}
\begin{proof}
    To prove the expression, we look at how quantum channel $\mathcal{T}$ acts on each of the three terms. 
    
    Firstly, since $\hat{P}_{\alpha} \hat{\Pi}_0  = \hat{\Pi}_0 \hat{P}_{\alpha} = 0, \, \forall \alpha$ and $\hat{\Pi}_0\, \hat{\rho}_0\, = \hat{\rho}_0\, \hat{\Pi}_0 = \hat{\rho}_0$, we have
    \begin{eqnarray*}
        \mathcal{T}({\hat{\rho}_0}) &=& \hat{\rho}_0 + \frac{2\beta}{m}\sum_{\alpha}\left[\hat{P}_{\alpha} \hat{\rho}_0\hat{P}_{\alpha} - \frac{1}{2}\left( \hat{P}_{\alpha} \hat{\rho}_0 + \hat{\rho}_0 \hat{P}_{\alpha} \right)\right] \\
        &=& \hat{\rho}_0 + \frac{2\beta}{m}\sum_{\alpha}\left[\hat{P}_{\alpha} \hat{\Pi}_0\hat{\rho}_0\hat{\Pi}_0\hat{P}_{\alpha} - \frac{1}{2}\left( \hat{P}_{\alpha} \hat{\Pi}_0\hat{\rho}_0 + \hat{\rho}_0 \hat{\Pi}_0 \hat{P}_{\alpha} \right)\right] = \hat{\rho}_0 \; .
    \end{eqnarray*}
    
    Secondly, for $\hat{\rho}_{\perp}$ we have
    \begin{equation}
    \begin{aligned}
        \mathcal{T}(\hat{\rho}_{\perp}) = \hat{\rho}_{\perp} + \frac{2\beta}{m}\sum_{\alpha}\left[\hat{P}_{\alpha} \hat{\rho}_{\perp}\hat{P}_{\alpha} - \frac{1}{2}\left( \hat{P}_{\alpha}\hat{\rho}_{\perp} + \hat{\rho}_{\perp}\hat{P}_{\alpha} \right)\right].
    \end{aligned}
    \end{equation}
    Since $\hat{P}_{\alpha} \hat{\Pi}_0  = \hat{\Pi}_0 \hat{P}_{\alpha} = 0$ and $\hat{\Pi}_0\hat{\rho}_{\perp} = \hat{\rho}_{\perp} \hat{\Pi}_0 =0$, we observe that $\hat{\Pi}_0  \mathcal{T}(\hat{\rho}_{\perp}) =  \mathcal{T}(\hat{\rho}_{\perp}) \hat{\Pi}_0 = 0$. This implies that $\mathcal{T}(\hat{\rho}_{\perp})$, like $\hat{\rho}_{\perp}$, has support only on the subspace orthogonal to $\hat{\Pi}_0$.
    
Thirdly, for $\hat{c}$, we have
    \begin{equation}
    \begin{aligned}
        \mathcal{T}(\hat{c}) &= \hat{c} + \frac{2\beta}{m}\sum_{\alpha}\left[\hat{P}_{\alpha} \hat{c} \hat{P}_{\alpha} - \frac{1}{2}\left( \hat{P}_{\alpha}\hat{c} + \hat{c} \hat{P}_{\alpha} \right)\right]\\
        &= \hat{c} + \frac{2\beta}{m}\sum_{\alpha}\left[\hat{P}_{\alpha} \hat{\Pi}_0 \hat{c} \hat{P}_{\alpha} - \frac{1}{2}\left( \hat{P}_{\alpha}\hat{\Pi}_0 \hat{c} + \hat{c} \hat{P}_{\alpha} \right)\right]\\
        &= \hat{c}\, (\hat{\mathds{1}}-\beta \hat{O}) \; ,
    \end{aligned}
    \end{equation}
with $\hat{\mathds{1}}$ denoting the identity. This implies that we have $\hat{\Pi}_0 \mathcal{T}(\hat{c}) = \mathcal{T}(\hat{c})$ and $\mathcal{T}(\hat{c})\, \hat{\Pi}_0 = 0$, as soon as the same are true for $\hat{c}$.

Combining the above three parts, we see that the corresponding blocks of $\hat{\rho}$ undergo independent dynamics under $\mathcal{T}(\hat{\rho})$, with $\hat{\rho}_0$ remaining constant and e.g.~$\gamma^\prime \hat{c}^\prime = \gamma \hat{c}\, (\hat{\mathds{1}}-\beta \hat{O})$. The latter implies
\begin{eqnarray*}
(\gamma^\prime)^2 \mathrm{Tr}[ \hat{c}^{\prime \dagger} \hat{\rho}_0^{(pi)} \hat{c}^{\prime}   ] & = & \gamma^2\, \mathrm{Tr}[\hat{c}^\dagger \hat{\rho}_0^{(pi)} \hat{c}\, (I-\beta \hat{O})^2 ] = \gamma^2\, \sum_i\; (\hat{c}^\dagger \hat{\rho}_0^{(pi)} \hat{c})_{(i,i)}\, (1-\beta \lambda_i)^2 \\
&\leq& \gamma^2 \, \left( \max_{i:\, (\hat{c}^\dagger \hat{\rho}_0^{(pi)}\hat{c})_{(i,i)}>0} (1-\beta \lambda_i)^2  \right) \; \sum_i\; (\hat{c}^\dagger \hat{\rho}_0^{(pi)}\hat{c})_{(i,i)} \\ &\leq& (\gamma (1-\beta G))^2\; \mathrm{Tr}[\hat{c}^\dagger \hat{\rho}_0^{(pi)} \hat{c}] \; ,
\end{eqnarray*}
as stated. Here $(\hat{X})_{(i,i)}$ denotes diagonal component $i$ of operator $\hat{X}$ in the eigenbasis of $\hat{O}$ and the first inequality follows from $\hat{\Pi}_0 \hat{c}^\dagger \hat{\rho}_0^{(pi)} \hat{c} \hat{\Pi}_0 = 0$, implying that $(\hat{c}^\dagger \hat{\rho}_0^{(pi)} \hat{c})_{(i,i)} =0$ for all $\lambda_i=0$.
\end{proof}

Lemma \ref{lemma_decomposition} basically says that the coherence between eigenspaces spanned by $\hat{\Pi}_0$ and by $(\hat{\mathds{1}}-\hat{\Pi}_0)$ decays at a rate governed by the spectral gap of $\hat{O}(\theta)$ when applying the channel $\mathcal{T}(\theta)$. Furthermore, the populations on these two eigenspaces remain constant.
Next, we prove how this decay of coherences limits the decrease in fidelity with respect to $\hat{\Pi}_0(\theta)$ upon changing $\theta$.

\mainproposition*
\begin{proof}
    From Lemma \ref{lemma_decomposition}, we have
    \begin{equation}
    \begin{aligned}
        f^{\prime} &= \mathrm{Tr}\left[\hat{\Pi}_0(\theta + \Delta \theta) \mathcal{T}(\hat{\rho}) \right] \\
        &= f\, \mathrm{Tr}\left[ \hat{\Pi}_0(\theta + \Delta \theta) \hat{\rho}_0\right] + (1-f)\, \mathrm{Tr}\left[ \hat{\Pi}_0(\theta + \Delta \theta) \hat{\rho}^{\prime}_{\perp}\right]\; + \; \gamma^\prime \, \mathrm{Tr}[\hat{\Pi}_0(\theta + \Delta \theta)\, (\hat{c}^\prime+\hat{c}^{\prime \dagger})]\\
        & \geq f\, \mathrm{Tr}\left[ \hat{\Pi}_0(\theta + \Delta \theta) \hat{\rho}_0\right]\;\; - \gamma [1-\beta \, G(\theta)]\; \left\vert \mathrm{Tr}[\hat{\Pi}_0(\theta+\Delta\theta)\,(\hat{c}^\prime + \hat{c}^{\prime \dagger})] \right\vert \\
        & \geq f \delta \;\; - \gamma [1-\beta \, G(\theta)]\; \left\vert \mathrm{Tr}\left[\hat{\Pi}_0(\theta+\Delta\theta)\,\left(\hat{\Pi}_0^{(r)}(\theta)\hat{c}^\prime(\hat{\mathds{1}}-\hat{\Pi}_0(\theta)) + (\hat{\mathds{1}}-\hat{\Pi}_0(\theta))\hat{c}^{\prime \dagger}\hat{\Pi}_0^{(r)}(\theta)\right)\right] \right\vert \; . 
    \end{aligned} 
    \end{equation}
We recall that $\Pi_0^{(r)} = \rho_0 \, \rho_0^{(pi)}$ is the projector onto the subspace of $\Pi_0$ spanned by $\rho_0$. Thanks to the Cauchy-Schwarz inequality $\vert \mathrm{Tr}[\hat{A} \hat{B}] \vert \leq \sqrt{\mathrm{Tr}[A^\dagger A] \mathrm{Tr}[B^\dagger B]}$, we can bound the first term inside the trace as follows:
\begin{equation}\label{eq:MultLineProof}
\begin{aligned}
\gamma \mathrm{Tr}& \left[\hat{\Pi}_0(\theta+\Delta\theta) \hat{\Pi}_0^{(r)}(\theta)\hat{c}^\prime(\hat{\mathds{1}}-\hat{\Pi}_0(\theta))\right]\\ &= \gamma \mathrm{Tr}\left[\hat{\Pi}_0(\theta+\Delta\theta)\, \sqrt{\rho_0} \sqrt{\rho_0^{(pi)}}\, \hat{c}^\prime(\hat{\mathds{1}}-\hat{\Pi}_0(\theta))\right] \\
&= \gamma \mathrm{Tr}\left[\left((\hat{\mathds{1}}-\hat{\Pi}_0(\theta)) \hat{\Pi}_0(\theta+\Delta\theta) \sqrt{\rho_0}\right)\left(\sqrt{\rho_0^{(pi)}} \hat{c}^\prime\right)\right] \\
&\leq \mathrm{Tr}\left[(\hat{\mathds{1}}-\hat{\Pi}_0(\theta))\; \hat{\Pi}_0(\theta+\Delta\theta)\; \rho_0\; \hat{\Pi}_0(\theta+\Delta\theta)\right]^{1/2} \, \gamma\mathrm{Tr}\left[\hat{c}^{\prime \dagger}\rho_0^{(pi)} \hat{c}^\prime\right]^{1/2} \\
&\leq \left( \sum_{i,j} p_i \,\left\vert \langle \lambda_i | \hat{\Pi}_0(\theta+\Delta\theta)| \lambda_j \rangle \right\vert^2 \right)^{1/2} \; \sqrt{f(1-f)} \; ,
\end{aligned}
\end{equation}
where $i$ spans the eigenvectors associated to the zero eigenspace of $\hat{O}(\theta)$ according to $\rho_0=\sum_i \, p_i |\lambda_i\rangle\langle \lambda_i|$ and $j$ spans the eigenvectors associated to the nonzero eigenvalues. Defining $| v_i \rangle = \hat{\Pi}_0(\theta+\Delta\theta) |\lambda_i\rangle$, we can just decompose
\begin{eqnarray*}
    |v_i \rangle = \sum_k | \lambda_k \rangle\langle \lambda_k |\hat{\Pi}_0(\theta+\Delta\theta)  | \lambda_i \rangle + \sum_j | \lambda_j \rangle\langle \lambda_j |\hat{\Pi}_0(\theta+\Delta\theta)  | \lambda_i \rangle
\end{eqnarray*}
where $i,k$ run over the 0 eigenvalues and $j$ over the nonzero eigenvalues. Then 
\begin{eqnarray*}
\langle v_i |v_i \rangle &=& \langle \lambda_i | \hat{\Pi}_0(\theta+\Delta\theta) |\lambda_i\rangle \\
&=& \sum_k |\langle \lambda_k |\hat{\Pi}_0(\theta+\Delta\theta)  | \lambda_i \rangle|^2 + \sum_j |\langle \lambda_j |\hat{\Pi}_0(\theta+\Delta\theta)  | \lambda_i \rangle|^2
\end{eqnarray*}
implies that 
\begin{eqnarray*}
\sum_{j} \,\left\vert \langle \lambda_i | \hat{\Pi}_0(\theta+\Delta\theta)| \lambda_j \rangle \right\vert^2 &=& \langle \lambda_i | \hat{\Pi}_0(\theta+\Delta\theta) |\lambda_i\rangle - \sum_k |\langle \lambda_k |\hat{\Pi}_0(\theta+\Delta\theta)  | \lambda_i \rangle|^2 \\
&\leq& \langle \lambda_i | \hat{\Pi}_0(\theta+\Delta\theta) |\lambda_i\rangle - |\langle \lambda_i | \hat{\Pi}_0(\theta+\Delta\theta) |\lambda_i\rangle|^2 \, .
\end{eqnarray*}
The function $x-x^2$ is positive, upper bounded by $1/4$ for $x\in [0,1]$ and decreasing for $x\in [1/2,1]$. Applying this to $x=\langle \lambda_i | \hat{\Pi}_0(\theta+\Delta\theta) |\lambda_i\rangle \geq \delta$ implies $\sum_{j} \,\left\vert \langle \lambda_i | \hat{\Pi}_0(\theta+\Delta\theta)| \lambda_j \rangle \right\vert^2 \leq a(\delta)$. Plugging this into the last line of \eqref{eq:MultLineProof} and summing over $i$ then yields the announced result.
\end{proof}

When $\mathcal{T}$ is applied $M$ times, the only modification in the proof is that $\gamma$ decreases to $\gamma' \leq (1-\beta G(\theta))^M$ by recursively applying Lemma \ref{lemma_decomposition}. This readily justifies the extension \eqref{eq:fidelity_drop2}.

\subsection{Convergence of Zeno dragging and the advantage of weak continuous limit}\label{appendix:convergence of zeno}

To ensure a limited drop in fidelity when dragging between $\theta=\theta_i$ and $\theta=\theta_f$, on the basis of \eqref{eq:fidelity_drop2}, we impose sufficiently small steps $\Delta\theta$ to maintain $\delta$ close to 1, and take $M$ sufficiently large to maintain $f' \sim f \delta $ at each step.

\maintheorem*
\begin{proof}
There are $N$ increments of $\theta$ in the algorithm, so we need to iterate \eqref{eq:fidelity_drop2} N times. Drop the factor of $f(1-f)$ on the right hand side of \eqref{eq:fidelity_drop2} (which is valid since $f(1-f)\leq 1$), assume $\delta \geq 1/2$ and assume $M(\theta)$ is chosen such that $[1-\beta G(\theta)]^{M(\theta)} \leq \alpha < 1$. Then \eqref{eq:fidelity_drop2} readily yields
\begin{eqnarray}
f_{k+1} + \alpha\, m(\delta) &\geq& \delta \, \left( f_k - 2 \alpha \sqrt{\delta-\delta^2}/\delta\right) + \alpha\, m(\delta) \\
&\geq & \delta (f_k + \alpha \, m(\delta)) \; ,
\end{eqnarray}
provided we choose $m(\delta) \geq 2 \frac{\sqrt{\delta-\delta^2}}{1-\delta}$; for instance, take $m(\delta) = \sqrt{2}\sqrt{\frac{1+\delta}{1-\delta}}$. This recursively yields
\begin{equation}
f_N \geq \delta^N f_0 - \sqrt{2}\, \alpha\,(1-\delta^N)\sqrt{\frac{1+\delta}{1-\delta}} \,  .
\end{equation}
Starting with $f_0 = 1$, we now impose parameters to ensure $\delta^N \geq 1-\epsilon_1$ and $\sqrt{2}\, \alpha\,(1-\delta^N)\sqrt{\frac{1+\delta}{1-\delta}} < \epsilon_2$ to conclude the proof.

Concerning the first requirement, the assumption translates into noting that $\delta \geq \left(\mathrm{cos}(\frac{\Delta \theta}{2})\right)^{2n} \geq 1-\frac{n\Delta \theta^2}{4}$, we then obtain 
$$\delta^N \geq \left(\mathrm{cos}(\frac{\Delta \theta}{2})\right)^{2nN} \geq 1 - \frac{N n\Delta \theta^2}{4} = 1 - \frac{n \phi^2}{4N}\geq 1-\epsilon_1$$ 
by choosing $N$ as stated.

Concerning the second requirement, we have
\begin{equation}
\begin{aligned}
    \sqrt{2}\alpha (1-\delta^N)\sqrt{\frac{1+\delta}{1-\delta}} &= \sqrt{2}\alpha \frac{(1-\delta^N)}{1-\delta}\sqrt{1-\delta^2} \\
    &= \sqrt{2}\alpha \sqrt{1-\delta^2} \sum_{k=0}^{N-1}\delta^k \\
    &\leq \sqrt{2}\alpha \sqrt{1 -(1-\tfrac{2n\phi^2}{4 N^2})} \sum_{k=0}^{N-1} 1\\
    &= \sqrt{2}\alpha \sqrt{\frac{n}{2}}\frac{\phi}{N}N \; = \; \alpha\sqrt{n} \phi \; .
\end{aligned}
\end{equation}
From there we readily see that taking $M(\theta)$ as stated allows us to choose $\alpha < \frac{\epsilon_2}{\sqrt{n}\phi}$ and thus to satisfy this second requirement.
\end{proof} 

The above theorem still contains a hidden dependence on time, inside the measurement strength $\beta$. The total time-to-solution is computed as follows.

\maincorollary*
\begin{proof}
    The total time is just the number $N = \frac{n\phi^2}{4\epsilon_1}$ of $\theta$ increment steps, times the number $M$ of repetitions of the channel $\mathcal{T}(\theta)$ at each $\theta$ value --- which we lower bound from the Theorem statement, replacing $G(\theta)$ by its lower bound $G_{\min}$ and $\beta=1-e^{-\Delta t/2\tau}$ --- , times the duration $\Delta t$ of a single application of $\mathcal{T}(\theta)$. 
\end{proof}

\noindent \emph{Remark:} Let us denote $\mathcal{T}_{\Delta t}$ the Kraus map associated to a particular choice of $\Delta t$. The above Corollary thus considers the scaling as a function of $w$ of applying, for every value of $\theta$, the Kraus map $(\mathcal{T}_{1/w})^w$. When measuring a single clause, it is easy to see from \eqref{eq:measurement_channel} that $(\mathcal{T}_{\Delta t})^2 \equiv \mathcal{T}_{2 \Delta t}$. Physically, this is a consequence of the commutation of consecutive weak measurements of the same observable; it implies that the total measurement accuracy versus total measurement time does not depend on the chosen value of $\Delta t$ in this case. Once we measure several non-commuting clauses, we have $(\mathcal{T}_{\Delta t})^2 \not\equiv \mathcal{T}_{2 \Delta t}$ in \eqref{eq:measurement_channel}. Physically, the advantage of $\Delta t \rightarrow 0$ can be understood as saying that the exact unraveling of the measurement process over time becomes important, once we claim to weakly simultaneously measure non-commuting observables.

\section{Analytically solvable optimal control: a single-qubit example }\label{appendix:analytical_single_qubit}
Here we study an analytically solvable model consisting of only a single qubit. We consider the system driven either by a Lindblad master equation for Lindblad-OFS problem
\begin{equation}
    \frac{d\hat{\rho}}{dt} = \Gamma \left( \hat{\sigma}(\theta) \, \hat{\rho} \, \hat{\sigma}(\theta) - \hat{\rho}\right),
\end{equation} 
or by a stochastic master equation for MLP-OFS problem
\begin{equation}
    \frac{d\hat{\rho}}{dt} = \sqrt{\Gamma}r\left[\hat{\sigma}(\theta)\hat{\rho} + \hat{\rho}\hat{\sigma}(\theta) - 2 \hat{\rho}\mathrm{Tr}(\hat{\sigma}(\theta)\hat{\rho}) \right],
\end{equation}
where $rdt = 2\sqrt{\Gamma}\mathrm{Tr}(\hat{\sigma}(\theta)\hat{\rho})dt + dW$ is the measurement readout and $dW$ is the white noise. Here $\Gamma$ is the measurement rate and the measurement observable $\hat{\sigma}(\theta)$ is the generalized Pauli operator at the basis determined by $\theta$
\begin{equation}
    \hat{\sigma}(\theta) = \mathrm{cos}(\theta)\hat{\sigma}_x + \mathrm{sin}(\theta) \hat{\sigma}_z.
\end{equation}
In both Lindblad-OFS and MLP-OFS problems, the initial state is $\hat{\rho}(0) = \frac{\hat{\openone} + \hat{\sigma}(\phi_i)}{2}$, and we want to drag the qubit to $\frac{\hat{\openone} + \hat{\sigma}(\phi_f)}{2}$ within the given dragging time $T_f$, where $\phi_i$ and $\phi_f$ are initial and target azimuthal angles of the qubit in the $xz$-plane of the Bloch sphere. We hence define a cost function 
\begin{equation}
    J_f = \mathrm{Tr}(\hat{\rho}(T_f) \hat{\sigma}(\phi_f))
\end{equation}
which only depends on the final state. The task is to find the optimal schedule $\theta(t)$ for $t \in [0, T_f]$ such that $J_f$ is maximized. 

The solution to MLP-OFS is derived in \cite{lewalle2024prxq}, and the optimal schedule $\theta^\star(t)$ is a linear function of time if the state is initially pure:
\begin{eqnarray}
    \theta^\star(t) = \phi_i + (\phi_f - \phi_i)\frac{t}{T_f} + \arctan\left(\frac{\phi_f - \phi_i}{4\Gamma T_f}\right).
\end{eqnarray}
Notice there is an offset $\arctan\left(\frac{\phi_f - \phi_i}{4\Gamma T_f}\right)$ between the control $\theta^\star(t)$ and the state's (supposed) azimuthal coordinate $\phi$. 

We now derive the solution to the Lindblad-OFS problem and show the optimal schedule $\theta^\star(t)$ is also a linear function of time, though with a different offset. Just like Lindblad-OFS, this problem can also be solved exactly within the framework of Pontryagin maximization principle (PMP), which we illustrate here. PMP defines a control Hamiltonian $\mathbb{H}_c$ which in this case can be calculated to be
\begin{equation}
\begin{aligned}
    \mathbb{H}_c &= \Gamma \Lambda_x\left[ R_x (\mathrm{cos}(2\theta)-1) + R_z \,\mathrm{sin}(2\theta)\right] \\
    &+\Gamma \Lambda_z[R_x \mathrm{sin}(2\theta) - R_z (\mathrm{cos}(2\theta)+1)],
\end{aligned}
\end{equation}
where we have parameterized $\hat{\rho} = \frac{\hat{\mathds{1}} + R_x \hat{\sigma}_x+R_z \hat{\sigma}_z}{2}$, and $\Lambda_x$ and $\Lambda_z$ are the conjugate variables with respect to the state variables $R_x$ and $R_z$. Using the canonical transformation 
\begin{equation}
    \begin{aligned}
R_x & \rightarrow R \cos \phi \\R_z & \rightarrow R \sin \phi \\
\Lambda_x & \rightarrow \Lambda_R \cos \phi-\Lambda_\phi \sin \phi / R \\
\Lambda_z & \rightarrow \Lambda_R \sin \phi+\Lambda_\phi \cos \phi / R
\end{aligned}
\end{equation}
to move from Cartesian to polar coordinates, we can rewrite the control Hamiltonian in the polar coordinate
\begin{equation}\label{eq:H_c_polar}
    \mathbb{H}_c = \Gamma \left\{R \Lambda_R\left[\cos(2(\theta - \phi)) - 1\right] + \Lambda_\phi \sin(2(\theta - \phi)) \right\}.
\end{equation}
The PMP states that the state variable $\mathbf{q}$ (for $ \mathbf{q}= R $ or $\phi$) satisfies
\begin{equation}\label{eq:Hamilton_eq_state}
    \dot{\mathbf{q}} = \frac{\partial \mathbb{H}_c}{\partial \Lambda_\mathbf{q}},
\end{equation}
and the conjugate variable $\Lambda_\mathbf{q}$ satisfies
\begin{equation}\label{eq:Hamilton_eq_costate}
     \dot{\Lambda}_\mathbf{q} = -\frac{\partial \mathbb{H}_c}{\partial \mathbf{q}}.
\end{equation}
Moreover, PMP further states that $\mathbb{H}_c$ is maximized over all possible values of $\theta$ at any time point $t$. For unbounded $\theta$, this is equivalent to 
\begin{equation}
    \frac{\partial \mathbb{H}_c}{\partial\theta} = 0,
\end{equation}
which evaluates to be 
\begin{equation}\label{eq:offset}
    \theta = \phi + \frac{1}{2}\tan^{-1}\left(\frac{\Lambda_\phi}{R\Lambda_R}\right).
\end{equation}
Notice that as $\mathbb{H}_c$ is only a function of $\theta - \phi$, $\frac{\partial \mathbb{H}_c}{\partial\theta} = -\frac{\partial \mathbb{H}_c}{\partial\phi} $ so $\frac{\partial \mathbb{H}_c}{\partial\phi} = 0$ as well. This means $\Lambda_\phi$ is a constant of motion. Also, the form of $\mathbb{H}_c$ in \eqref{eq:H_c_polar} gives rise to 
\begin{equation}
\begin{aligned}
    \frac{d(R\Lambda_R)}{dt} = \dot{R}\Lambda_R + R \dot{\Lambda}_R = \Lambda_R \frac{\partial \mathbb{H}_c}{\partial \Lambda_R} - R \frac{\partial \mathbb{H}_c}{\partial R} = 0,
\end{aligned}
\end{equation}
which means $R\Lambda_R$ is also a constant of motion. Therefore, the offset $\theta - \phi = \frac{1}{2}\tan^{-1}\left(\frac{\Lambda_\phi}{R\Lambda_R}\right)$  between the control $\theta$ and the state variable $\phi$ is a constant. 

The equation of motion for state variable $\phi$ gives the rate of change for $\phi$ 
\begin{equation}
    \dot{\phi} = \frac{\partial \mathbb{H}_c}{\partial \Lambda_\phi} = \Gamma \sin(2(\theta - \phi)),
\end{equation}
which is a constant. This means $\phi$ is a linear function of time, and so is the control $\theta$ due to \eqref{eq:offset}
\begin{equation}
    \phi(t) = \phi_i + (\phi(T_f) - \phi_i) \frac{t}{T_f},
\end{equation}
and also we can rewrite $\dot{\phi}$ as 
\begin{equation}\label{eq:rewrite_rate_1}
    \dot{\phi} = \Gamma \sin(2(\theta - \phi)) = \frac{\Gamma\Lambda_\phi}{\sqrt{R^2\Lambda_R^2 + \Lambda_\phi^2}} = \frac{\phi(T_f) - \phi_i}{T_f}.
\end{equation}
The PMP also provides the terminal condition for the conjugate variables
\begin{equation}
\begin{aligned}
    \left.\Lambda_\phi(T_f) = - \frac{\partial J_f}{\partial \phi}\right|_{T_f} = - R(T_f)\sin(\phi_f - \phi(T_f)) \\
    \left.\Lambda_R(T_f) = - \frac{\partial J_f}{\partial R}\right|_{T_f} = - \cos(\phi_f - \phi(T_f)),
\end{aligned}
\end{equation}
which leads to $\cot(\phi_f-\phi(T_f)) = \frac{R \Lambda_R}{\Lambda_\phi} $. Here we used $J_f = R(T_f) \cos(\phi_f - \phi(T_f))$.  This means we can further rewrite \eqref{eq:rewrite_rate_1} as 
\begin{equation}\label{eq:algebraic_eq}
    \sin(\phi_f - \phi(T_f)) = \frac{\phi(T_f) - \phi_i}{\Gamma T_f}.
\end{equation}
Solving $\phi(T_f)$ from this with given $T_f$, the optimal schedule then is determined to be 
\begin{equation}
    \theta(t) = \phi_i + (\phi(T_f) - \phi_i)\frac{t}{T_f} + \frac{1}{2}(\phi_f - \phi(T_f))
\end{equation}
With the above results, we can also solve for $r(t)$ given optimal schedule from 
\begin{equation}
    \dot{R} = \frac{\partial \mathbb{H}_c}{\partial \Lambda_R} = \left(\frac{R \Lambda_R}{\sqrt{R^2 \Lambda_R^2 + \Lambda_\phi^2}} -1 \right)\Gamma R = -[1-\cos(\phi_f -\phi(T_f))]\Gamma R,
\end{equation}
from which we have $R(T_f) = e ^{-[1-\cos(\phi_f -\phi(T_f))]\Gamma T_f}$. Therefore the optimal final cost under the optimal schedule is 
\begin{equation}
    J_f^{*} = e ^{-[1-\cos(\phi_f -\phi(T_f))]\Gamma T_f} \cos(\phi_f -\phi(T_f)),
\end{equation}
where $\phi(T_f)$ is determined from \eqref{eq:algebraic_eq}, and it can be shown this is a monotonic increasing function of $T_f$, consistent with adiabatic intuition.

\section{Equation reduction steps}\label{appendix:Equation_reduction}
In Section \ref{sec_OCT} of the main text, we derived the equation of motion \eqref{eq:eom_costate_mlp_opt_final} that the costate $\hat{\Lambda}$ has to satisfy under the optimality condition. In this appendix, we show that the equation can be further simplified into the negative conjugate of \eqref{eq:eom_rho_mlp_opt_final}. Consider the SME
\begin{equation}
    \frac{\partial \hat{\rho}}{\partial t}=\hat{F}\hat{\rho}+\hat{\rho}\hat{F},
\end{equation}
with 
\begin{equation}
    \hat{F} = \frac{1}{m\tau}\sum_\alpha(r_\alpha-1)\left(\hat{P}_\alpha-\left\langle\hat{P}_\alpha\right\rangle\right),
\end{equation}
and 
\begin{equation}
     f = \frac{1}{m\tau}\sum_\alpha(r_\alpha-1)\left\langle\hat{P}_\alpha\right\rangle.
\end{equation}
Then we can write the costate evolution as 
\begin{equation}
    \frac{\partial \hat{\Lambda}}{\partial t}= -\hat{F}\hat{\Lambda}-\hat{\Lambda}\hat{F}+2(\left\langle\hat{\Lambda}\right\rangle-1)\left(\hat{F}+f\right).
\end{equation}
It is easy to see that
\begin{equation}
     \frac{\partial \left\langle\hat{\Lambda}\right\rangle}{\partial t}=(2\left\langle\hat{\Lambda}\right\rangle-1)f.
\end{equation}
We now redefine $\hat{\Lambda}^\prime=\hat{\Lambda}+\left(1-\expval{\hat{\Lambda}}\right)\hat{\mathds{1}}$. Then, it can be seen that 
\begin{equation}
    \frac{\partial \hat{\Lambda}^\prime}{\partial t}= -\hat{F}\hat{\Lambda}^\prime-\hat{\Lambda}^\prime\hat{F}.
\end{equation}
We can now  redefine our costate to be $\hat{\Lambda}^\prime.$ We can similarly change the transversality condition. Note, this transformation does not change the CDJ stochastic Hamiltonian.

\section{Connection between costate operator and CDJ variables}\label{appendix:costate}

In this section, we analyze how the costate operator  $\hat{\Lambda}$, defined in sec.~\ref{OC_costate}, is connected to the costate vector $\{\mathbf{\Lambda}\}$. For notational convenience, we denote the latter as $\{\mathbf{p}\}$ in this section. Recall, for a continuously monitored system, the CDJ stochastic path integral expresses the probability densities of a specific trajectory $\{\mathbf{q},\mathbf{r}\}$ and is given by \cite{CJ} \begin{equation}
     \mathcal{P}=\int\mathcal{D}[\mathbf{p}]e^\mathcal{S},
\end{equation}
where $\mathbf{q}$ provides a parametrization of the system state (e.g.~Bloch coordinates for a qubit, position and momentum expectation values for an oscillator) and $\mathbf{p}$ are the corresponding costate variables.
Now, we will see how the costate operator $\hat{\Lambda}$ connects to such costate variables for a single qubit. For a single qubit
\begin{eqnarray}
    \hat{\rho}=\frac{1}{2}\left(\hat{\mathds{1}}+x\hat{\sigma}_x+y\hat{\sigma}_y+z\hat{\sigma}_z\right),
\end{eqnarray}
where $(x,y, z)\equiv (q_1, q_2, q_3)$ are the Bloch coordinates of the qubit. Any costate operator $\hat{\Lambda}$ can be expressed in terms of the Pauli operators as below
\begin{eqnarray}
    \hat{\Lambda}=\hat{\mathds{1}}+\Lambda_x\left(\hat{\sigma}_x-x\hat{\mathds{1}}\right)+\Lambda_y\left(\hat{\sigma}_y-y\hat{\mathds{1}}\right)+\Lambda_z\left(\hat{\sigma}_z-z\hat{\mathds{1}}\right).
\end{eqnarray}
In the above expression, we used the frame transformation $\hat{\Lambda}\equiv \hat{\Lambda}+\left(1-\expval{\hat{\Lambda}}\right)\hat{\mathds{1}}$ and $\{\Lambda_{x},\Lambda_{y},\Lambda_{z}\}$ are scalars.
If the qubit evolves according to $\dot{\hat{\rho}}=\mathcal{F}[\hat{\rho}]$, we can find the Bloch coordinate evolution by $\dot{q}_j=f_j=\textrm{tr}\left(\,\hat{\sigma}_j\mathcal{F}\right)$. Now, consider PMP based optimization of the cost
\begin{eqnarray}
    J = \int_0^{t_f}dt\,f^0(\hat{\rho},\chi),
\end{eqnarray}
with control $\chi$. Here, $f^0$ is an arbitrary function. Then, the Pontryagin Hamiltonian becomes
\begin{eqnarray}
    \mathcal{H}=-f^0+\sum_jp_jf_j.
\end{eqnarray}
However, note that $\textrm{tr}\,\left(\hat{\Lambda}\mathcal{F}\right)=\sum_j\Lambda_jf_j$. Thus, we can also write 
\begin{eqnarray}
    \mathcal{H}=-f^0+\textrm{tr}\left(\hat{\Lambda}\mathcal{F}\right),
\end{eqnarray}
with $\Lambda_j = p_j$. Hence, the CDJ costate variables can be interpreted as a parametrization of the costate operator.  For a general quantum system if $q_j = \textrm{tr}\left(\hat{X}_j\hat{\rho}\right)$, then, we can express, $\hat{\Lambda}=\sum_jp_j\hat{X}_j$. Thus, the costate operator constrains the different degrees of freedom in the system.  On the other hand, if $\hat{\Lambda}=\sum_j\lambda_j\hat{\Pi}_j$ is an eigen decomposition, then $\{\hat{\Pi}_j\}$ represent effective degrees of freedom of motion.

To interpret the path integral of the form  $\int\mathcal{D}[\hat{\Lambda}]$, we can parametrize both $\hat{\rho}$ and $\hat{\Lambda}$, in terms of the state coordinates and corresponding costate variables and perform path integrals as usual \cite{CDJ_origin,CJ}.

\clearpage
\newpage
\bibliographystyle{quantum}
\bibliography{main.bib}
\end{document}